\documentclass[11pt,a4paper]{article}

\usepackage{amssymb,amsmath, amsfonts}

\usepackage{stix}
\usepackage{graphicx,graphics}
\usepackage{mathtools}
\usepackage[english]{babel}
\usepackage[utf8]{inputenc}
\usepackage{csquotes}
\usepackage{epsfig,url}
\usepackage{bbm}
\usepackage{theorem}
\usepackage{a4wide}
\usepackage{enumerate}
\usepackage{verbatim} 
\usepackage{color}
\usepackage{esint}

\usepackage[T1]{fontenc} 

\usepackage{hyperref}

\newcommand{\ud}{\;\mathrm{d}}
\newcommand{\uud}{\mathrm{d}}

	\providecommand{\vr}{{\vartheta_r}}
	
	\providecommand{\loc}{\mathnormal{loc}}

\providecommand{\hPhi}{\hat{\phi}}
\providecommand{\s}{\mathcal{S}}

\providecommand{\ol}{\overline}

\providecommand{\eps}{\varepsilon}

\providecommand{\La}{\mathcal{L}}
\providecommand{\F}{\mathcal{F}}
\providecommand{\Haus}{\mathcal{H}}
\providecommand{\PV}{\operatorname{PV}}

\providecommand{\G}{\mathcal{G}}
\providecommand{\V}{\mathcal{V}}

\DeclareMathAlphabet{\pazocal}{OMS}{zplm}{m}{n}
\newtheorem{theorem}{Theorem}[section]
\newtheorem{definition}[theorem]{Definition}
\newtheorem{lemma}[theorem]{Lemma}

\newtheorem{notation}[theorem]{Notation}
\newtheorem{assumption}[theorem]{Assumption}

{\theorembodyfont{\upshape}
\newtheorem{remark}[theorem]{Remark}

}
\numberwithin{equation}{section}
\numberwithin{theorem}{section}
\newcommand{\qed}{\hfill$\Box$}
\newenvironment{proof}{\begin{trivlist}\item[]{\em Proof:}\/}{\qed\end{trivlist}}
\newenvironment{proofof}[1][Proof]{\noindent \textit{#1.} }{\ \qed}
\DeclareMathOperator{\sign}{\mathrm{sign}}

\newcommand{\Reals}{{\mathbb R}}
\newcommand{\Complex}{{\mathbb C\hspace{0.05 ex}}}
\newcommand{\Naturals}{{\mathbb N}}

\newcommand{\cf}{{\mathbbm 1}}

\newcommand{\Sch}{\mathcal{S}}

\title{The two-particle correlation function for systems with long-range interactions}

\author{ Juan J. L. Vel\'azquez
\thanks{\emailjuan}, Raphael Winter \thanks{\emailalessia}  \\[1em]
$\,^\ddag$\UBaddress}
\date{\today}
\newcommand{\email}[1]{E-mail: \tt #1}
\newcommand{\emailjuan}{\email{velazquez@iam.uni-bonn.de}}
\newcommand{\emailalessia}{\email{raphaelwinter@iam.uni-bonn.de}}
\newcommand{\UBaddress}{\em University of Bonn, Institute for Applied Mathematics, Bonn,  Germany}

\date{\today}
\begin{document} 
 
\maketitle
\begin{abstract}
	In this paper, we study the truncated two-particle correlation function in particle systems with
	long range interactions. For Coulombian and soft potentials, we define and give well-posedness results for the
	equilibrium correlations. In the Coulombian case, we prove the onset of the Debye screening length in the equilibrium correlations, for suitable velocity distributions. Additionally, we give precise estimates on the effective range of interaction
	between particles. In the case of soft potential interaction the equilibrium correlations and their fluxes in the space of velocities are shown to
	be linearly stable. 
\end{abstract}

\tableofcontents
\newpage

\section{Introduction}

\subsection{Kinetic limits of particle systems with long-range interactions}

A classical problem studied in statistical physics is the dynamics of systems of many identical particles
which interact by means of long range potentials. In particular, this problem has received a big deal of 
attention in the community working on plasma physics in the case in which particles interact via the Coulomb
potential.

Early contributions to this topic were made by Bogolyubov \cite{bogoliubov_problems_1962}, and  have been extended by 
the works of Balescu \cite{balescu_equilibrium_1975,balescu_statistical_1963}, as well as Guernsey \cite{guernsey_kinetic_1962}
and Lenard \cite{lenard_bogoliubovs_1960}. These authors obtained a kinetic equation which describes the behavior of the velocity distribution of
a spatially homogeneous many particle system with long range interaction (in particular Coulomb forces). Bogolyubov derived the following system of equations for the density $f_1(\tau,v_1)=f_1(\tau,x_1,v_1)=f_1(\tau,\xi_1)$, rescaled truncated correlation function $\tilde{g}_2(\tau,x_1,v_1,x_2,v_2)=\tilde{g}_2(\tau,\xi_1,\xi_2)$, and a small parameter $\sigma>0$ tending to zero:
\begin{align}
\partial_\tau f_1 &= \sigma \nabla_v \cdot \left(\int \nabla \phi(x_1-x_3) \tilde{g}_2(\xi_1,\xi_3) \ud{\xi_3} \right)\label{BBGKYint3f} \\
\partial _\tau \tilde{g}_2+&\sum_{i=1}^2 v_i \nabla_{x_i} \tilde{g}_2- \sum_{i=1}^{2}\int \nabla \phi (x_i-x_3)\nabla_{v_i} f_1(\tau,\xi_i)\tilde{g}_2(\xi _{\zeta \left( i\right) },\xi_{3}) \ud{\xi_3} \label{BBGKYint3g} \\
&=  (\nabla _{v_1}-\nabla_{v_2})\left(f_1(\tau,\xi_1)f_1(\tau,\xi_2)\right)\nabla \phi (x_1-x_2)\notag.
\end{align} 
Here $\phi$ is the interaction potential, and $\zeta(1)=2,\zeta(2)=1$ exchanges the variables. 
Actually, \cite{bogoliubov_problems_1962} derives analogous approximations for higher
order correlations, but those are of lower order in $\sigma \rightarrow 0$.
In this paper, we will consider two classes of potentials, namely
the Coulomb potential $\phi(x)=\frac{c}{|x|}$ for some $c>0$, and so-called soft potentials, that are radially symmetric functions
in the Schwartz class. 

In order to find the limit equation for $f_1$ as $\sigma\rightarrow 0$, Bogolyubov argues that all terms in \eqref{BBGKYint3g} are of the same order of magnitude, so the evolution of $\tilde{g}_2$ can be observed in times of order one.
Since $\tilde{g}_2$ is of order one, it can be expected that $f_1$  evolves on the longer timescale $t=\sigma \tau $. We assume that for $f_1$ given, $\tilde{g}_2$ has a globally stable equilibrium. We will call the steady state equation
\begin{equation}
\begin{aligned} \label{Bogolyubov}
\sum_{i=1}^2 v_i \nabla_{x_i} g_B- \sum_{i=1}^{2} \nabla_{v_{i}} f_1  \int \nabla \phi (x_i-x_3)g_B(\xi _{\zeta \left( i\right) },\xi_{3})\ud{\xi_3} 	 
=  (\nabla_{v_1} - \nabla_{v_2})\left(f_1 f_1\right)\nabla \phi (x_1-x_2).	
\end{aligned}
\end{equation}
the Bogolyubov equation and the solution $g_B$ the (truncated) Bogolyubov correlation.
In the paper \cite{bogoliubov_problems_1962}, it is argued that the equation \eqref{Bogolyubov} should be solved subject to the boundary condition:
\begin{align} \label{eq:bogboundary}
g_B(x-\tau v_1,v_1,x_2-\tau v_2,v_2) \rightarrow 0 \quad \text{ as $\tau\rightarrow \infty$ }.
\end{align} 
This condition can be interpreted as particles being uncorrelated before they come close enough to interact.
Then we can immediately predict the limiting kinetic equation for $f_1$ on the timescale $t$ by plugging $\tilde{g}_2=g_B$ into \eqref{BBGKYint3f}. This yields the Balescu-Lenard equation:
\begin{align} \label{eq:kinetic}
\partial_t f(t,v) 	&= \nabla_v \cdot \left( \int_{\Reals^3} a(v-v',v) (\nabla_v-\nabla_{v'})(f(t,v)f(t,v')) \ud{v'}\right)\\
a_{i,j} (w,v)		&=  \int_{\Reals^3} k_ik_j \delta(k\cdot w) \frac{|\hat{\phi}(k)|^2}{|\eps(k,k \cdot v)|^2} \ud{k}.		
\end{align}
Here, $\eps$ is the so-called dielectric function, which we introduce in Definition \ref{def:dielectric}.
We remark that the integral defining $a$ is logarithmically divergent for large values of $k$ in the case of Coulomb interaction. We will discuss
this in detail in Subsection \ref{sec:coulkin}.
The equation \eqref{eq:kinetic} shares many properties with classical kinetic equations like the Boltzmann equation and the Landau equation. In particular,
the steady states of \eqref{eq:kinetic} are the Maxwellian distributions:
\begin{align}\label{def:M}
	M(v):=  \left(\frac{m}{2\pi k_B T}\right)^\frac32 e^{- \frac{m|v|^2}{2k_BT}}.
\end{align}
Moreover, the entropy $H[f(t,\cdot)]= - \int f(t,v) \log(f(t,v)) \ud{v}$ of a solution $f$ of \eqref{eq:kinetic} is (formally) increasing in time, as remarked in \cite{lenard_bogoliubovs_1960}. 

The Balescu-Lenard equation \eqref{eq:kinetic}, was found independently by Guernsey \cite{guernsey_kinetic_1962} and Lenard (cf. \cite{lenard_bogoliubovs_1960}), following
the approach by Bogolyubov, and along a different line by Balescu (cf. \cite{balescu_statistical_1963}). There are also stochastic derivations of the
Balescu-Lenard equation using different arguments, which are discussed in Subsection \ref{subsec:Screen}.

The first characterization of the solution to the steady state equation \eqref{Bogolyubov} has been obtained by Lenard in \cite{lenard_bogoliubovs_1960}, yielding a formal derivation of the Balescu-Lenard equation \eqref{eq:kinetic}. The Lenard approach, which is based on a Wiener-Hopf argument, yields an explicit formula for the right-hand side of \eqref{BBGKYint3f}, when $\tilde{g}_2$ is a steady state of \eqref{BBGKYint3g} with $f_1$ fixed.  A Fourier representation of the full steady state $g_B$ was found later by Oberman and Williams  \cite{oberman_theory_1983} using a similar approach. There are few rigorous results on the Balescu-Lenard equation \eqref{eq:kinetic}. The linearized equation has been studied in \cite{strain_linearized_2007}. 

The results presented in this paper are the following. First we study the well-posedness of \eqref{Bogolyubov}. Secondly, we study the stability properties
of the steady state $g_B$ under the evolution given by \eqref{BBGKYint3g} for fixed $f_1$. Thirdly, we study the decay properties of the steady
states $g_B$. The steady state $g_B$ encodes the information on the range of interaction of particles within the system. To understand this, consider two particles at phase space positions  $\xi_j=(x_j,v_j)$, $j=1,2$. Let $b(\xi_1,\xi_2)$ be the impact	parameter, and $d(\xi_1,\xi_2)$ be the signed distance of the first particle to the collision point. More precisely, the impact parameter $b$ is defined as the vector from $x_2$ to $x_1$ at their time of closest approach along the free trajectories, so $b$ and $d$, (and the negative part $d_-$) are given by:
\begin{align} \label{def:impactdist}
b(\xi_1,\xi_2) = P^\perp_{v_1-v_2}(x_1-x_2), \quad d(\xi_1,\xi_2) = (x_1-x_2)\cdot \frac{v_1-v_2}{|v_1-v_2|}, \quad d_- = \max\{0,-d\}. 
\end{align}
We show that the function $g_B$ encodes a characteristic length scale emerging in the system, the so-called Debe-length $L_D$ (cf. \eqref{def:Debye}). In equation
\eqref{Bogolyubov}, this length has been rescaled to one.
The correlation of particles that remain at a distance much larger than the Debye length, i.e. $|b|\gg 1$, is expected to be negligible. Moreover, one expects negligible correlations for particles that (so far) have remained at a distance
larger than the characteristic length, that is $d_{-} \gg 1$.   
In this paper, we prove that for Coulomb interacting systems, the equilibrium correlations $g_B$ satisfy the following estimate, for every compact set $K\subset \Reals^3$ and $\delta>0$
\begin{align}\label{eq:intrscreenC}
|g_B(\xi_1,\xi_2)| &\leq \frac{C(\delta,K)}{|v_1-v_2|} \frac{1}{(|b|+d_-) (1+|b|+d_-)^{\gamma-\delta}}, \quad   \text{$v_1,v_2 \in K$}.
\end{align}
Here $\gamma=0$ if $f_1(v)$ decays exponentially, and $\gamma=1$ if $f_1$ 
behaves like a Maxwellian for large velocities. We observe that the result only shows the onset of a characteristic length scale, when the one-particle function $f_1$ behaves like a Maxwellian for large velocities, but not for exponentially decaying functions, indicating that a characteristic length in the system can only be expected for functions $f_1$ with Maxwellian decay.

We further note that \eqref{eq:intrscreenC} indicates that the correlations become singular for particles
with small impact factor $b$. This is crucial for identifying the kinetic equation for Coulombian particle systems and is discussed in Subsection  \ref{sec:coulkin}. 

In the case of soft potential interaction, we prove that the equilibrium correlations $g_B$ satisfy the estimate \eqref{eq:intrscreenC} with $\gamma=2$, even if the
potential decays exponentially. In this case, we do not observe a singularity for $|b|,|d_-|\rightarrow 0$.

A fact that will play a crucial role in the proof of \eqref{eq:intrscreenC} for the Coulomb potential are the zeros of the function $\Re(\eps(k,u))$ for $k\rightarrow 0$ ($\eps$ as in \eqref{eq:kinetic}), for which  $\Im(\eps(k,u))$ is exponentially small. These zeros are well-known in the physics literature, and related to the so-called Langmuir waves
(cf. \cite{lifshitz_course_1981}). These are plasma density waves with very large wavelength which damp out only very slowly. This is the physical cause 
for the slow Landau damping of Maxwellian plasmas. More precisely, it has been shown in \cite{glassey_time_1995,glassey_time_1994}  that the rate of convergence to equilibrium is only logarithmic in time for Maxwellian plasmas, that is when $f_1$ is a Maxwellian.
Furthermore, the zeros of $\Re(\eps(k,u))$ are crucial to the analysis of the linearized Balescu-Lenard equation in \cite{strain_linearized_2007}. In our paper, they account for the dependence of the screening properties (cf. \eqref{eq:intrscreenC}) on the behavior of the one-particle function for large velocities.

We study the linearized evolution of the truncated correlation function $\tilde{g}_2$ \eqref{BBGKYint3g} with fixed one-particle function. 
Similar to the Vlasov equation, the equation can be solved in Fourier-Laplace variables (cf. \cite{krommes_two_1976}). We introduce in Definition \ref{def:Bogolyubovprop}
the representation of the solution in terms of Vlasov propagators,
and in Section \ref{sec:stability} we show linear stability
of the Bogolyubov steady states $g_B$
\begin{align} \label{introd:g2conv}
	\tilde{g}_2(\tau,\cdot) \longrightarrow g_B(\cdot) \quad \text{in $D'(\Reals^9)$ as $\tau\rightarrow \infty$},
\end{align}
as well as stability of the fluxes on the right-hand side of
\eqref{BBGKYint3f}, for soft potentials $\phi$. The result \eqref{introd:g2conv} can be understood as a linear Landau damping result for two particles.

We remark that the reduction of the evolution problem to Vlasov equations stresses the importance
of a good understanding of the Vlasov-Poisson equation, in particular the stability of steady states. In the articles \cite{glassey_time_1995,glassey_time_1994} it is proved that solutions of the linear Vlasov-Poisson equation converge to spatially homogeneous states, however the result is restricted to the case of initial data that
are rotationally symmetric in the velocity variable.
On a one-dimensional periodic spatial domain, the spectral theory of the linearized Vlasov equation has been studied in \cite{degond_spectral_1986}. 
Due to the shortcomings of the current stability theory of the linear Vlasov-Poisson equation, the rigorous stability results for the truncated correlations $\tilde{g}_2$ in this work are obtained for soft potentials.

We now recall, in a more modern language, the main ideas in the original derivation of the system
\eqref{BBGKYint3f}-\eqref{BBGKYint3g} proposed by Bogolyubov. An overview over particle models and scaling limits in kinetic theory can be
gained from \cite{spohn_kinetic_1980,spohn_large_2012,villani_review_2002}.

Consider a system of particles $\{(\tilde{X}_j,\tilde{V}_j)\}_{j\in J}$ with unitary mass, where $J$ is a countable index set and $\tilde{X}_j,\tilde{V}_j \in \Reals^3$ denote
the position and velocity of particles. Let the evolution
of the system be given by:
\begin{align} \label{eq:Hamilton}
	\partial_\tau \tilde{X}_i(\tau) &= \tilde{V}_i(\tau), \quad \partial_\tau V_i = -\tilde{\sigma} \sum_{j\neq i} \nabla \phi(\tilde{X}_i-\tilde{X}_j).
\end{align}  
The parameter $\tilde{\sigma}$ can be interpreted as the squared charge of an individual particle and will be passed to zero later. 
We will assume that the initial configuration of particles
is random and distributed according to a spatially homogeneous Poisson point process
with an average of $\tilde{N}=\tilde{\sigma}^{-\kappa}$ particles per unit of volume for some $\kappa>0$. More precisely, the process has the intensity measure $\lambda(\uud{x}\uud{v})=\tilde{N} f_0(x,v) \uud{x}\uud{v}$, where $f_0(x,v)=f_0(v)$ is some probability density in the space of velocities.

The average kinetic energy of a particle, that we also call the temperature of the system, we will denote by $T$.
By rescaling velocities and time we can assume without loss of generality that $T=1$. 
We consider scaling limits of \eqref{eq:Hamilton} and try to characterize the statistical behavior of \eqref{eq:Hamilton} depending on the choice of the parameter $\kappa>0$ that determines
the interdependence of $\tilde{\sigma},\tilde{N}$, as well as the interaction potential $\phi$. 

In spite of the fact that the Coulomb potential does not have an intrinsic length scale, a characteristic length emerges from the dynamics of the system. To this end, we observe
that there are two independent quantities with the unit of a length that can be obtained from the 
quantities $\tilde{\sigma},\tilde{N}$ and $T$ describing the system. One of them is the typical distance of particles $d= \tilde{N}^{-\frac13}$. The second is the so-called Debye screening length:
\begin{align}\label{def:Debye}
	L_D = \sqrt{\frac{T}{\tilde{N}\tilde{\sigma}}} , 
\end{align}
which is well-known in plasma physics. 
Note that the definition \eqref{def:Debye} makes sense without a well-defined temperature, using the average kinetic energy instead of the temperature. We assume the average momentum of particles is zero. This way of defining the Debye-length is widely used in plasma physics for systems far away from thermal equilibrium, see for example \cite{lifshitz_course_1981}.
The Debye length will play a crucial role in many results of this paper. It measures the 
characteristic (effective) range of interaction between the particles of the system, assuming that the velocity distribution of particles
$f_1(v)$ satisfies a suitable stability condition (cf. Assumption \ref{Ass:onepart1}). Under this assumption, $L_D$ is the effective radius 
of a single particle, that is the characteristic distance to which the influence of a single particle can be felt in a system evolving according to \eqref{eq:Hamilton},
when $\phi$ is the Coulomb potential.  We can assume $L_D=1$ using the
change of variables:
\begin{align} \label{eq:CoulombUnits}
	L_D X &= \tilde{X}, \quad L_D \tau = \tilde{\tau}, \quad L_D  \theta^2 = \tilde{\theta}^2, \quad N= L_D^3 \tilde{N}.
\end{align}
After changing units, the average number of particles per unit volume $N$ and the rescaled strength $\sigma$ of the potential satisfy the relation:
\begin{align} \label{eq:coulthetasq}
	N  = \sigma^{-1}, \quad \sigma \rightarrow 0
\end{align}
and the particle system $\{(X_j,V_j)\}_{j \in J}$ satisfies \eqref{eq:Hamilton} with $\tilde{\sigma}$ replaced by $\sigma$. Hence, for systems evolving according to \eqref{eq:Hamilton} with $\phi$ the Coulomb potential, we can assume without loss of generality
that \eqref{eq:coulthetasq} holds. Therefore, we will consider particle
system determined by the scaling limit \eqref{eq:coulthetasq}, and compare the case of Coulomb interaction and the case of interaction with a smooth, decaying potential.

Let $\phi$ be a soft potential with characteristic length $\ell=1$. Then per unit
of time, a typical particle will interact with $N$ many particles and each interaction yields a deflection of order $\sigma$ with zero average.
If the forces of all particles within the range of the potential are independent, the variance of the sum of the deflections is:
\begin{align} \label{eq:var}
	\operatorname{Var}(V(\tau)) \sim \sigma \tau.
\end{align}
Therefore, the variance will become of order one on a macroscopic time scale $t=\sigma \tau$.

We are interested in the correlation of particles in the scaling limit of particle systems 
given by \eqref{eq:Hamilton}, \eqref{eq:coulthetasq}. The presentation will be similar to the one in \cite{velazquez_non-markovian_2018}.
Denote phase space variables by $\xi=(x,v)$, let 
$F_n(\tau,\xi_1,\ldots,\xi_n)$ be the $n$-particle correlation function of the system, and $f_n =  F_n/N^n$ be the rescaled correlation function. Formally, these functions satisfy the BBGKY hierarchy (cf. \cite{balescu_equilibrium_1975}). In the scaling limit \eqref{eq:coulthetasq}, the hierarchy reads as:
\begin{equation} \label{eq:BBGKY}
\begin{aligned}
	\partial_\tau f_n + \sum_{i=1}^n v_i \nabla_{x_i} f_n &- 
	\sum_{i=1}^n \int    \nabla \phi(x_i-x_{n+1}) \nabla_{v_i}
	f_{n+1} \ud{\xi_{n+1}}  	\\
	= &\sigma \sum_{i \neq j} \nabla \phi(x_i-x_j) \nabla_{v_i} f_n .	
\end{aligned}
\end{equation}
Since we assume that particles are initially independently distributed, the correlation functions
at the initial time $\tau=0$ factorize: $f_n(0,\xi_1,\ldots,\xi_n)= f_1(0,\xi_1)\cdots f_1(0,\xi_n)$.
The evolution given by \eqref{eq:Hamilton} will create correlations between particles. In order
to be able to study this, we introduce the (rescaled) truncated correlation functions $g_n$:
\begin{equation}	
	\begin{aligned}\label{def: g2g3}
		g_2(\xi_1,\xi_2) &= f_2(\xi_1,\xi_2) - f_1(\xi_1) f_1(\xi_2),\\
		g_3(\xi_1,\xi_2,\xi_3) &= f_3(\xi_1,\xi_2,\xi_3) - (f_1f_1 f_1)(\xi_1,\xi_2,\xi_3) \\
		&- f_1(\xi_1) g_2(\xi_2,\xi_3) - f_1(\xi_2) g_2(\xi_1,\xi_3) - f_1(\xi_3) g_2(\xi_1,\xi_2), \\
		 \ldots .
	\end{aligned} 
\end{equation}  
Rewriting the equations BBGKY hierarchy \eqref{eq:BBGKY} in terms of the functions $g_n$ we find that 
a consistent assumption on the orders of magnitudes is:
\begin{align}
	g_n \approx \sigma^{n-1}.
\end{align} 
Hence we expect that, to leading order, the equations for $f_1$, $g_2$ (cf. \eqref{eq:BBGKY}) can be approximated by:
\begin{equation} \label{eq:truncBBGKY}
	\begin{aligned}
	\partial_\tau f_1 &=  \nabla_v \cdot \left(\int \nabla \phi(x_1-x_3) g_2(\xi_1,\xi_3) \ud{\xi_3} \right) \\
	\partial _\tau g_2+& \sum_{k=1}^2 v_k \nabla_{x_k} g_2 -\sum_{k=1}^{2}\int \nabla \phi (x_k-x_3)\nabla_{v_{k}}(f_{1}(\xi_k)g_2(\xi _{\zeta \left( k\right) },\xi_{3})) \ud{\xi_3} 	\\
		&=\sigma\sum_{k=1}^{2}\nabla _{v_{k}}\left(f_{1}(\xi_1)f_{1}(\xi_2)\right)\nabla \phi (x_k-x_{\zeta(k)}).
	\end{aligned}
\end{equation}
Since the source term for $g_2$ in \eqref{eq:truncBBGKY} is of order
$\sigma$, the function $\tilde{g}_2= \sigma^{-1}g_2$ can be expected to be of order one. With this definition, \eqref{eq:truncBBGKY} is equivalent to \eqref{BBGKYint3f}-\eqref{BBGKYint3g}.

In scaling limits with weak interaction, e.g. the weak-coupling limit, one can apply a similar reasoning. In this case,  steady state equation for the truncated correlations is 
\begin{equation}
\begin{aligned} \label{BogolyubovLandau}
\sum_{i=1}^2 v_i \nabla_{x_i} g_B =  (\nabla _{v_1} - \nabla_{v_2})\left(f_{1}f_{1}\right)\nabla \phi (x_1-x_2).	
\end{aligned}
\end{equation} 
Notice that the integral term in \eqref{Bogolyubov} disappears in the case of weak interaction.
The equation \eqref{BogolyubovLandau} can be solved explicitly using the method of characteristics. In this case the resulting kinetic equation
for $f_1$ is formally the Landau equation. Partial results on the derivation can be found in \cite{bobylev_particle_2013,velazquez_non-markovian_2018}. Global well-posedness and stability for the Landau equation has been proved in \cite{guo_landau_2002}.

We then summarize the main implications of the results for the study of scaling limits of Coulomb particle systems. 
Most importantly, the Debye screening
becomes visible in the length scale of the two-particle correlation function \eqref{eq:intrscreenC}. It is worth mentioning that the different decay exponents $\gamma$ in the result suggests that the screening properties depend on the behavior of the one-particle function $f_1$ for large velocities. The Debye screening can also be observed on the level of the linearized Vlasov equation. We will take a closer look at this in Subsection \ref{subsec:Screen}.
  
Further, the argument identifies two regions in which
the assumption $f_1\gg g_2$ breaks down, namely for particles $\xi_1$, $\xi_2$ with very small relative velocity $v_1-v_2\approx 0$, and
very fast particles. The critical region of particles with very small relative velocity is a result of the fact that the collision time diverges, when
particles only very slowly separate (see \cite{velazquez_non-markovian_2018}). 

A mathematical description of scaling limits of Coulomb particle systems requires to understand the following aspects:
Firstly, the emergence of the Debye length $L_D$ from the particle system \eqref{eq:Hamilton}. Secondly, one needs
to estimate the deflections due to the interaction of particles with an impact parameter much larger than the Debye length. Due to the
screening, the influence of a single charge decays much faster than the Coulomb potential itself. Thirdly, one needs to understand
the deflections produced by particles that approach closer than the Debye length. The influence of these deflections turns out to be dominant
by a logarithmic factor $\log(\frac1{\sigma})$ and yields the Landau equation in the kinetic limit. This is discussed in Subsection~\ref{sec:coulkin}.

\subsection{Debye Screening in the Vlasov equation} \label{subsec:Screen}

In this subsection, we discuss the onset of a screening length in the linearized Vlasov equation.
To this end, we will take a closer look at the steady states of the Vlasov-Poisson equation in the presence
of a point charge. The Debye screening can be observed in the decay of the equilibrium spatial profile, which
has a characteristic length scale that is given by the Debye length $L_D$ (cf. \eqref{def:Debye}), in spite of the fact that the Coulomb potential
does not have a length scale. The screening effect is related to the classical subjects in 
the Vlasov theory such as Landau damping and Langmuir waves (cf. \cite{glassey_time_1995,glassey_time_1994,lancellotti_glassey-schaeffer_2015,lifshitz_course_1981,mouhot_landau_2011,penrose_electrostatic_1960}).

 We prove in this paper, that the evolution problem
\eqref{BBGKYint3g} can be reduced to the Vlasov system. We remark that one can also formally derive the Balescu-Lenard equation from a stochastic model involving Vlasov equations.
The method consists in describing the evolution of the probability density of a tagged particle which interacts with a random medium. The random medium is assumed to evolve according to the Vlasov equation, linearized around the velocity distribution of the tagged particle.
The approach of a Vlasov medium is well-studied in the formal theory in plasma physics \cite{piasecki_stochastic_1987, rostoker_superposition_1964}. 
Rigorous results on a related model can be found in \cite{lancellotti_fluctuations_2009,lancellotti_time-asymptotic_2016}. 

Let $(X,V)$ be the phase space coordinates of the tagged particle traveling through a
continuous background, with which it interacts via the Coulomb potential. Here $f_0(v)$ is a fixed velocity distribution,
and $h(\tau,x,v)$ the correction that is induced by the particle. 
Taking as unit of length the Debye length $L_D$ (cf. \eqref{def:Debye}) as before, let the system be given by: 
\begin{align}
		\partial_\tau h + v \nabla_x h -  \nabla_x (\phi * \varrho) \nabla_v f_0 	&= \sigma  \nabla_v f_0  \nabla\phi(x-X(\tau)),  &h(0,x,v)&= 0 \label{eq:Vlasovcharge} \\
												\varrho(x)			&= \int h(x,v) \ud{v} \label{eq:Vlasovdensity} \quad &&{}\\
							\partial_\tau X			= V,\quad		&\partial_\tau V			= -\sigma\nabla_x (\phi * \varrho)(X(\tau)),  &(X(0),V(0))&=(X_0,V_0) \label{eq:ODE}.
\end{align}  
In the derivations of the Balescu-Lenard equation in \cite{lancellotti_fluctuations_2009,lancellotti_time-asymptotic_2016,piasecki_stochastic_1987}, the initial datum $h(0,\cdot)$ in \eqref{eq:Vlasovcharge} is random. 
Then the dynamics describing the evolution of $(X,V)$ becomes a stochastic differential equation. Notice that
the evolution of random measures under the Vlasov equation has already been considered in Braun and Hepp (cf. \cite{braun_vlasov_1977}).
In the system \eqref{eq:Vlasovcharge}-\eqref{eq:ODE}, $(X,V)$ can be interpreted as a particle traveling through a random background of particles, and
$h(x,v)$, $\varrho(x)$ as the correction of the homogeneous density (or "cloud") induced by the particle.
It is worth noting that the well-posedness of the problem of a moving point charge interacting with a fully nonlinear Vlasov-Poisson system has been studied in \cite{desvillettes_polynomial_2015}.

For simplicity, assume $f_0(v)$ in \eqref{eq:Vlasovcharge} is radially symmetric. In the derivation of the Landau equation and the Balescu-Lenard equation, we make the assumption that the trajectories of particles are approximately rectilinear on the microscopic timescale. This suggests to approximate $X(\tau)$ in \eqref{eq:Vlasovcharge} by
\begin{align} \label{eq:rectilinear}
	X(\tau) \approx X_0 - \tau V_0.
\end{align}
For the special case $V_0=0$, it was observed in \cite{lifshitz_course_1981} that the Debye screening can be derived from the equation \eqref{eq:Vlasovcharge}.
The spatial density of the steady state of \eqref{eq:Vlasovcharge} with a point charge at rest can be computed explicitly (without loss of generality $X_0=0$):
\begin{align} \label{eq:VlasovsolRest}
	\varrho_{eq}(x) = \frac{\sigma}{4\pi |x|} e^{-|x|}.
\end{align}
Remarkably, even though the potential $\phi(x)=1/|x|$ does not have a length scale, the spatial profile of $\varrho_{eq}$ decays exponentially with characteristic
scale given by the Debye length $L_D$.  

Now consider the case of $V_0 \neq 0$. Making the assumption of rectilinear motion \eqref{eq:rectilinear}, we can again solve \eqref{eq:Vlasovcharge} explicitly.
For $\tau \rightarrow \infty$, the solution converges to traveling wave with velocity $V_0$. The spatial profile of the traveling wave can
be represented in Fourier variables. Let $f_0$ be a given one-particle function, then the formula reads:
\begin{align} \label{eq:Vlasovsoltraveling}
	\hat{\varrho}_{trav}(k) = \frac{\sigma \int \frac{k \nabla f_0(v)}{k(v-V_0)-i0}\ud{v}}{|k|^2 D(k,k\cdot V_0)},
\end{align}
where $D(k,u)$ is given by:
\begin{align} \label{eq:Ddielectric}
	D(k,u):= 1 - \frac{1}{|k|^2}  \int_{\Reals^3} \frac{ k \cdot \nabla f_0(v)}{k\cdot v- u+i0}\ud{v}.
\end{align}
We remark that \eqref{eq:Ddielectric} suggests that for $|V_0|\rightarrow \infty $, the spatial profile $\varrho_{trav}(x)$ can have large oscillations with long wavelength $\lambda=1/|k|\rightarrow \infty$. To see this, we decompose $D=D_R + iD_I$ into its real and imaginary part. For $|k|\rightarrow 0$ and $u$ of order one, we have
the asymptotic formula
\begin{align} \label{eq:DRDI}
	D_R(k,u) \sim 1-1/|u|^2, \quad  D_I(k,u) = 1/|k|^2 \int_{k\cdot v=u} k/|k| \nabla f_0(v) \ud{v}. 
\end{align}
Hence, the real part of $D$ in \eqref{eq:Ddielectric} has a zero for $|k|\rightarrow 0$,  $u \sim 1$, and the imaginary part depends on the
tail behavior of the one-particle function $f_0$. This suggests that the traveling wave $\varrho_{trav}$ (cf. \eqref{eq:Vlasovsoltraveling}) surrounding
the particle $(X,V)$ can lead to large deflections in other particles for $|V_0|\gg 1$, depending on the decay of $f_0(v)$ for large velocities. 
In the presence of very fast particles, the rectilinear approximation \eqref{eq:rectilinear} does not hold. However, this should not affect the validity of the final
kinetic equation in the limit $\sigma \rightarrow 0$, since the number of particles with velocity $|V_0|\gg 1$ becomes negligible.

This observation explains
why the exponent in the estimate \eqref{eq:intrscreenC} depends on the decay properties of the one-particle functions, and the estimate is only valid for velocities varying on a compact set.

The zero of the real part $D_R$ (cf. \eqref{eq:DRDI}) is also related to other important phenomena in plasma physics, such as the so-called Langmuir waves.
The length of the Langmuir waves is much larger than the Debye length and the oscillation frequency has been normalized to $\Omega_{\operatorname{Langmuir}}=1$ in our setting. The amplitudes of these waves decrease exponentially at a rate proportional to $D_I$ (cf. \eqref{eq:DRDI}), so the rate strongly depends on the background
distribution of particles. For a Maxwellian distribution of particles $f_0=M$, the imaginary part is exponentially small, which results in a very slow Landau damping
as observed in \cite{glassey_time_1995,glassey_time_1994}.

\subsection{On the range of validity of the Balescu-Lenard equation for Coulomb potentials} \label{sec:coulkin}

The goal of this subsection is to determine the correct kinetic equation for scaling limits
of particle systems interacting with the Coulomb potential, or the Coulomb potential
smoothed out at the origin. It was already remarked by Lenard in \cite{lenard_bogoliubovs_1960},
that the integral \eqref{eq:kinetic} is not well-defined for $\phi(x)=1/|x|$, since the integral
\begin{align} \label{eq:diffcoeff}
	a_{i,j} (w,v)	&=  \int_{\Reals^3} k_ik_j \delta(k\cdot w) \frac{|\hat{\phi}(k)|^2}{|\eps(k,k \cdot v)|^2} \ud{k}
\end{align} 
is logarithmically divergent for large $k$. This corresponds to the divergence \eqref{eq:intrscreenC} for small values of the spatial
variable $x$, so the main contribution comes from the singularity of the Coulomb potential at the origin.
 
In the scaling limit \eqref{eq:coulthetasq}, particle interaction is given by the potential $\sigma \phi(x)= \sigma/|x|$. 
Therefore, an interaction of particles with impact parameter $|b|\leq \sigma$ will result in a deflection of order one.
This yields a Boltzmann collision term in the limit equation, as observed in \cite{nota_theory_2018}. We now
analyze the influence of interactions with impact parameter $|b|\geq \sigma$. This corresponds to
a truncation $\tilde{a}_{i,j}$ of the integral \eqref{eq:diffcoeff} to $|k|\leq \sigma^{-1}$.
 As Lenard observed in \cite{lenard_bogoliubovs_1960}, the function $\eps(k,k\cdot v)\rightarrow 1$ becomes
constant for $k\rightarrow \infty$. Therefore, the truncated coefficient $\tilde{a}$ satisfies:
\begin{equation} \label{logterm} 
\begin{aligned}
	\tilde{a}_{i,j} (w,v)	&= \lim_{\sigma \rightarrow 0}  |\log(\sigma)|\int_{B_{\sigma^{-1}}}  \frac{k_ik_j \delta(k \cdot w)|\hat{\phi}(k)|^2}{|\eps(k,k \cdot v)|^2} \ud{k} 
						\sim   \delta_{i,j} -  \frac{w_i w_j}{|w|^2}.	
\end{aligned}
\end{equation}   
Hence, we obtain the Landau kernel in this limit. Now we discuss how this observation connects to \eqref{BBGKYint3f}-\eqref{BBGKYint3g} for $\sigma \rightarrow 0$. Due to \eqref{logterm}, the kinetic timescale is not given by
$t=\sigma \tau$, but slightly shorter by a logarithmic correction. Therefore, the mathematically rigorous kinetic equation associated
to the scaling limit \eqref{eq:coulthetasq} is expected to be the Landau equation, and the
main contribution is due to the interaction of particles with very small impact factor.  
However a more accurate description of physical systems might be obtained by keeping the terms of the order
$|\log(1/\sigma)|^{-1}$ in the equation, since in physical systems, $|\log(1/\sigma)|$ cannot be expected to be very large (cf. the discussion in §41 of \cite{lifshitz_course_1981}). 
Therefore, the physical equation describing plasmas can be expected to involve a Balescu-Lenard term, the Landau collision operator and a Boltzmann collision operator. The relative size of the different collision terms would depend on the physical system in question. The Balescu-Lenard equation is the correct  limit equation
for systems with soft potential interaction in the scaling limits \eqref{eq:coulthetasq}.
    
Consider particle systems interacting via the Coulomb potential and take as unit of length the Debye length $L_D$ \eqref{def:Debye}.
As a simplified problem, one can study a smooth variant of the Coulomb potential, that is  $\phi_{C,r} \in C^\infty$ radially symmetric and $\phi_{C,r}(x)=1/|x|$ for $|x|\geq 1$. Then the kinetic equation associated to the scaling limit \eqref{eq:coulthetasq} can be expected to be the Balescu-Lenard equation. Notice that the equation includes the screening effect, that is expected since $\phi_{C,r}(x)$ coincides with the Coulomb potential for large $|x|$.   
   
A characterization of the limit equations for scaling limits of Lorentz models with long-range interaction (i.e. a tagged particle in a random, but fixed, background of scatterers) can be found in \cite{nota_theory_2018}. For mathematical results in this direction see also \cite{desvillettes_linear_1999, marcozzi_derivation_2016}. 
\section{Preliminary and main results}

\subsection{Definitions and assumptions}

For future reference we fix the notation for some classical integral transforms.

\begin{notation}
	We will use the following conventions for the Laplace
	transform $\La(f)$, the Fourier transform $\hat{f}$ and the 
	Fourier-Laplace transform $\tilde{f}$:
	\begin{align} 
	\La(f)(z) 		&= \int_0^\infty e^{-zt} f(t) \ud{t} \label{def:LapT} \\
	\F(f)(k)=\hat{f}(k)		&= \frac{1}{(2\pi)^\frac{n}2}\int_{\Reals^n} f(x)e^{-ix \cdot k} \ud{x} \label{def:FourT} \\
	\tilde{f}(z,k)	&= \frac{1}{(2\pi)^\frac32}\int_{\Reals^3} \int_0^\infty  f(t,x) e^{-zt} e^{-ix \cdot k} \ud{t} \ud{x}  \label{def:FLT}.
	\end{align}
\end{notation}
\begin{definition} \label{def:Pdefs}
	We define operators $P^+$, $P^-$ and $P$ on $L^2(\Reals)$, that
	on Schwartz functions $f\in \Sch(\Reals)$ are given by:
	\begin{align}
	P^\pm[f](x) 	&:= \lim_{\delta \rightarrow 0^+} \int_{\Reals} \frac{f(x')}{x'-x\mp i \delta} \ud{x'}, \label{def:Ppm}\quad 
	P[f](x)			:= \PV \int_{\Reals}   \frac{f(x')}{x'-x} \ud{x'} 
	\end{align}
	where the principal value  integral $\PV$ is defined as: $\PV \int \ud{x'} = \lim_{\delta \rightarrow 0^+} \int  \cf(|x-x'|\geq \delta) \ud{x'}$. 
\end{definition}

\begin{notation}[Relative velocity and impact parameter] \label{Not:Feps}
	For vectors $k,v_1,v_2 \in \Reals^3$, $v_1\neq v_2$, $k\neq0$, we will use the following shorthand notation:
	\begin{align}\label{eq:vecshort}
		\omega = \frac{k}{|k|},\quad  v_r = v_1-v_2, \quad \vr= \frac{v_r}{|v_r|}.
	\end{align} 
	The impact parameter $b\in \Reals^3$ and the distance to the collision point $d\in \Reals$ of particles $(x_1,v_1)$, $(x_2,v_2)$ with relative position $x=x_1-x_2$ and relative velocity $v_r=v_1-v_2$ is defined as:
	\begin{align}
		d(x,v_r) = \frac{x \cdot v_r}{|v_r|} , \quad b(x,v_r) = x-P_{v_r}(x) = x- \frac{v_r(x \cdot v_r)}{|v_r|^2} .\label{def:Proj}
	\end{align}
\end{notation}
Due to the translation invariance of the system, the truncated correlation function $g_2(x,v,x',v')$ is a function of $x-x',v,v'$ only. By a slight abuse of notation, we
identify $g_2$ with the function:
\begin{align} \label{translationsystem}
g_2(x-x',v,v')=g_2(x,v,x',v').	
\end{align}
Also the function should be invariant under exchanging the two particles, so we impose the symmetry:
\begin{align} \label{particlesymmetry}
g_2(x,v,v')= g_2(-x,v',v).
\end{align}
This symmetry we include in the space of functions in which we solve the Bogolyubov equation.
\begin{definition}
	Define the functionals  $|h|[g]$, $h[g]$ given by the following formulas:
	\begin{align}
		|h|[g]= \int |g(x,v_1,v_2)| \ud{v_2}, \quad h[g] = \int g(x,v_1,v_2) \ud{v_2}.
	\end{align}
	Let $W$ be the function space given by:
	\begin{align} \label{def:W}
		W= \{g\in L^1_{loc}(\Reals^9): \text{ \eqref{particlesymmetry} holds, $|h|[g]\in L^1_{loc}$, $\sup_{|v|\leq R} \|h[g](\cdot,v)\|_{L^2}\leq C(R)$ for $R>0$} \}.
	\end{align}
\end{definition}
We now give a definition of a solution to the Bogolyubov equation.
We recall the space $L^1+L^2$ of functions $\zeta$ that can be decomposed
as $\zeta=\zeta_1+\zeta_2$ with $\zeta_1\in L^1$, $\zeta_2\in L^2$. 
\begin{definition}[Bogolyubov correlation] \label{def:weaksol}
	Let $\nabla \phi \in L^1+ L^2$, and $f \in W^{1,1}(\Reals^3)\cap W^{1,\infty}(\Reals^3)$ be a probability density.
	We say $g_B\in W$ is a solution to the Bogolyubov equation if for all $\psi \in C^\infty_c(\Reals^9)$
	\begin{equation}
	\begin{aligned} \label{eq:bogolyubov}
	-&\int (v_1-v_2) g_B \partial_x \psi 	-  \int \nabla f(v_1)  \nabla \phi(x+y)
	h[g_B](y,v_2) \psi(x,v_1,v_2)  \\
	&-  \int \nabla f(v_2) \nabla \phi(-x+y)h[g_B](y,v_1) \psi(x,v_1,v_2)
	= \int (\nabla_{v_1}-\nabla_{v_2})[f \otimes f ]\nabla \phi(x) \psi,
	\end{aligned} 
	\end{equation}
	and it satisfies the Bogolyubov boundary condition
	\begin{align} \label{bogcond}
	g_B(x-\tau(v_1-v_2),v_1,v_2) \rightarrow 0,\quad \text{as $\tau\rightarrow \infty$, a.e.}
	\end{align}
\end{definition}

\begin{definition}[Radon transform and dielectric function] \label{def:dielectric}
		Let $f\in L^1(\Reals^3)\cap L^\infty(\Reals^3)$. We define the Radon transform 
		$F: \Reals^3 \times \Reals \rightarrow \Reals$ associated to $f$ by ($\omega=\omega(k)$ as in \eqref{eq:vecshort}):
		\begin{align} \label{def:F}
		F(k,u):= \int_{\{v: \,\omega \cdot v=u\}} f(v) \ud{v}.
		\end{align}
		Further we define the dielectric function $\eps: \Reals^3 \times \Reals \rightarrow \Reals$ associated to $f\in W^{1,1}(\Reals^3)\cap W^{1,\infty}(\Reals^3)$ and a potential $\phi$ by:
		\begin{align} \label{def:eps}
		\eps(k,-|k|u) := 1-\hat{\phi}(k)P^-[\partial_u F(k,\cdot)](u).
		\end{align}
		Here the operator $P^-$ defined in \eqref{def:Ppm} is applied in the second variable of $\partial_u F$.
		As a shorthand we also introduce the functions $\alpha$, $\alpha^-$ given by:
		\begin{align}
			\alpha(\chi,u)	:= P[\partial_u F(\chi,\cdot)](u), \quad \alpha^-(\chi,u)	:= P^-[\partial_u F(\chi,\cdot)](u) \label{def:alpha}.
		\end{align}
\end{definition}
\begin{remark}
	Note that the dielectric function $\eps$ coincides with the function $D$ introduced in \eqref{eq:Vlasovsoltraveling}, which quantifies the correction
	to the homogeneous density induced by a single point charge.
\end{remark}
The following definitions will be useful in studying the linear evolution problem \eqref{BBGKYint3g} for $g$.
When $f$ is time independent, the equation \eqref{BBGKYint3g} for $g$ can be solved explicitly. 
To this end we introduce some notation.
\begin{notation}
	We introduce the function: 
	\begin{align} \label{def:Q}
	Q(k,v) & = k \nabla f(v) \hPhi(k).
	\end{align}
	Furthermore, for a function $h(x,v)$ and a potential $\phi$ we set $E_h$ to be the self-consistent
	potential associated to $h$:
	\begin{align} \label{def:E}
		E[h](x)=E_h(x) = \int \int \phi(x-y) h(y,v) \ud{v} \ud{y}.
	\end{align}
\end{notation}
\begin{definition}[Vlasov and transport propagator] \label{def:Vlasovprop}
	Let $\phi$ be a radially symmetric Schwartz potential.
	Let $\V$ be the linear Vlasov propagator associated to $f$, so let $%
	\V(t)[h_{0}]=h(t)$ be the solution to:
	\begin{equation}
	\begin{aligned} \label{linVlasov}
	\partial_{t} h + v \nabla_{x} h - \nabla E_{h} \nabla f & =0,\quad 	h(0,\cdot) & = h_{0}(\cdot),
	\end{aligned}
	\end{equation}
	with $E_h$ as in \eqref{def:E}. In Fourier-Laplace variables (cf. \eqref{def:FLT}) the solution is given by:
	\begin{align}
	\tilde{h}(z,k,v)  = \frac{\hat{h}_{0}(k,v)}{z+ikv}+ \frac{i Q(k,v) \tilde{%
			\varrho}(z,k)}{z+ikv}  \label{hvlas}, \quad 
	\tilde{\varrho}(z,k) = \frac{\int\frac{\hat{h}_{0}(k,v^{\prime})}{%
			z+ikv^{\prime}}\;\mathrm{d}{v^{\prime}}}{\eps(k,-iz)}, 
	\end{align}
	with $Q$ as introduced in \eqref{def:Q}. Further let $T$ be the free transport propagator so 
	\begin{align}\label{eq:transportprop}
	T(t)[g](\xi_{1},\xi_{2}) := g(x-v_{1}t,v_{1},x_{2}-v_{2}t,v_{2}).
	\end{align}
\end{definition}	
\begin{definition} \label{def:Bogolyubovprop}
	Let $\tilde{g}_0(\xi_1,\xi_2)=g_0(x_1-x_2,v_1,v_2)$, $g_0 \in \Sch((\Reals^3)^3)$ be symmetric in exchanging the variables $\xi_1$, $\xi_2$, and set $S(\xi_1,\xi_2)=\delta(\xi_1-\xi_2) f(v_1)$. We define the
	Bogolyubov propagator $\G$ by: 
	\begin{align}\label{eq:Bogolyubovprop}
	\G(t)[\tilde{g}_0] & := \V_{\xi_1}(t) \V_{\xi_2}(t)[S +  \tilde{g}_0 ] - T(t) [S] ,
	\end{align}
	where $\V_{\xi_1}$ is the Vlasov propagator acting the set of variables $(x_1,v_1)=\xi_1$, and $\V_{\xi_2}$ the
	propagator acting on $(x_2,v_2)=\xi_2$.
\end{definition}

We will analyze the equilibrium two-particle correlations for so-called soft potentials
and the Coulomb potential. Notice that we restrict our attention to radially symmetric
potentials.
\begin{assumption}[Potentials] \label{def:potentials}
	Let $\phi_C \in C(\Reals^3 \setminus \{0\})$ be the Coulomb potential, so
	$\phi_C(x)=\frac{c}{|x|}$ for some $c>0$. Assume without
	loss of generality that $c=\sqrt{\frac{\pi}{2}}$, when $\hat{\phi}(k)=\frac{1}{|k|^2}$. We say $\phi_S=\phi_S(|x|)$ is a soft potential if $\phi_S \in \Sch(\Reals^3)$.	
\end{assumption}
On the one-particle distribution function $f$ we make the following 
regularity assumptions.
\begin{assumption}[Regularity and Decay] \label{ass:regdecay}
	Let $f\in C^8(\Reals^3)$ be nonnegative and
	\begin{align} \label{Ass:decayf}
	|\nabla^m f(v)| \leq C e^{-|v|}, \quad \text{for $m=0,1,\ldots,8$}.
	\end{align}
	Further let $f$ be normalized to: 
	\begin{align} \label{Ass:fint} 
	\int f(v) \ud{v} = 1.
	\end{align} 	
\end{assumption}
Our proof of existence of Bogolyubov correlations requires the plasma to be stable. 
This can be mathematically formulated in terms of the dielectric function $\eps$ (cf. \eqref{def:eps}) associated to $f$. 

\begin{assumption}[Plasma stability]\label{Ass:onepart1}
	We say $f$ is stable if for all $k \in \Reals^3$, $\chi \in S^2$, $u\in \Reals$ we have:
	\begin{align} 
		|k|^2 \neq P^-[\partial_u F(\chi,\cdot)](u), \quad \text{in particular }|\eps(k,u)| &\neq 0, \quad \text{$\eps$ as in \eqref{def:eps}}. \label{fstable1}
	\end{align}
\end{assumption}
\begin{remark}
	The physical relevance of this condition is discussed in \cite{lifshitz_course_1981}.
	A necessary and sufficient condition for stability (cf. \eqref{fstable1}) was given by Penrose in \cite{penrose_electrostatic_1960}. 
	For example the condition \eqref{fstable1} is satisfied by functions $f$, for which $F(u)$ has precisely one maximum and no other critical points.  
\end{remark}

In order to prove (exponential) linear stability of the equilibrium correlations and their
fluxes we make a stronger analytic stability assumption on the plasma, which requires 
that we can extend the dielectric function to a strip in the complex plane.

\begin{assumption}[Strong plasma stability] \label{Ass:onepart2}
	Let $f>0$ be a Schwartz probability density on ${\mathbb{R}}^{3}$. 
	Let $F$ be the  Radon transform defined in \eqref{def:F} and
	$\phi=\phi_S$ a soft potential.
	Assume that there exists $c>0$ such that for all $\chi\in S^2$, $F(\chi,iz)$ has a holomorphic extension to the strip  $H_{c} := \{z\in \Complex: |\Re(z)|\leq c\}$ and on $H_{c}$ satisfies the estimate 
	\begin{align} \label{eq:Cauchy}
		|F(\chi,iz)| \leq\frac{C}{1+\Im(z)^{2}}.
	\end{align}
	We will assume that the associated extension of the dielectric function $z \mapsto \eps(k,-i|k|z)$ to the shifted
	right half-plane $H_{-c}^{-}:=\{z \in \Complex: \Re(z) \geq -c\}$  is bounded below uniformly: 
	\begin{align}
		|\eps(k,-i|k|z)|\geq c_0 >0, \quad \text{for $0\neq k\in \Reals^3$, $z\in H^-_{-c}$}. \label{ass:diel}
	\end{align}
\end{assumption}
We now introduce some technical assumptions, that we later use to quantify the rate of decay of the equilibrium correlations.
We distinguish functions $f$ that behave like an exponential as $|v|\rightarrow \infty$, specified in Assumption \ref{Ass:Exponential}, and functions
that behave like Gaussians, as specified in Assumption~\ref{Ass:Maxwellian}.
\begin{notation} \label{not:psidef}
	We recall the function $\alpha$ introduced in \eqref{def:alpha}.
	For $k\in \Reals^3$, $\chi\in S^2$, let $u_0^+(k,\chi)>0$, $u_0^+(k,\chi)<0$ be the solutions to:
	\begin{align} \label{def:u0}
	|k|^2 - \alpha(\chi, u_0^\pm) = 0,
	\end{align}	
	whenever \eqref{def:u0} has a unique solution with the prescribed sign. 
	Further write $I(k,\chi)$ for the set
	\begin{align} \label{eq:Idef}
		I(k,\chi)=(u_0^-(k,\chi)-1,u_0^-(k,\chi) +1) \cup (u_0^+(k,\chi)-1,u_0^+(k,\chi) +1).
	\end{align}
	Let $L^\pm(k,\chi)$, $\Psi^\pm(k,\chi,y)$ be given by:
	\begin{align}
		L^\pm(k,\chi)&=\frac{\partial_u F(\chi,u_0(k,\chi))}{\partial_u \alpha(\chi,u_0(k,\chi))}
		, \quad &\text{for $k\in \Reals^3$, $\chi\in S^2$}, \label{eq:Ldef} \\
		\Psi^\pm(k,\chi,y) &= u_0(k,\chi)+ y \frac{\partial_u F(\chi,u_0(k,\chi))}{\partial_u \alpha(\chi,u_0(k,\chi))}, \quad &\text{for $k\in \Reals^3$, $\chi\in S^2$, $y\in \Reals$}.
	\label{def:Psidef}
	\end{align}
\end{notation}
\begin{assumption}[Asymptotically exponential behavior] \label{Ass:Exponential} 
	Let $f$ satisfy the Assumptions \ref{ass:regdecay}-\ref{Ass:onepart1}. 
	Let $L^\pm=L^\pm(k,\chi)$ and $\Psi^\pm$ be as in Notation \ref{not:psidef}. We say $f$ behaves 
	asymptotically like an exponential if it satisfies the following for some $r,c,C>0$:
	\begin{align}
	|\nabla^6_{k,\chi,y} \left(\frac{|k|^3}{\partial_u \alpha(\chi,\Psi^\pm)}\right)| &\leq C, \quad \text{ for $|k|\leq r$, $\chi \in S^2$, $|y|\leq 
		 {L^\pm}^{-1}$}, \label{eq:exp1} \\
	 |\nabla^6_{k,\chi,y} \left(\frac{|k|^2-\alpha(\chi,\Psi^\pm)}{y\partial_u F(\chi,\Psi^\pm)}\right)| &\leq C, \quad \text{ for $|k|\leq r$, $\chi \in S^2$, $|y|\leq 
		 {L^\pm}^{-1}$},\label{eq:exp2a}	\\
	|\left(\frac{|k|^2-\alpha(\chi,\Psi^\pm)}{y\partial_u F(\chi,\Psi^\pm)}\right)|&\geq c, \quad \text{ for $|k|\leq r$, $\chi \in S^2$, $|y|\leq 
	{L^\pm}^{-1}$},\label{eq:exp2b} \\
	|\nabla^6_{k,\chi,y} \left(\frac{F(\chi,\Psi^\pm)}{\partial_u F(\chi,\Psi^\pm)}\right)| &\leq C, \quad \text{ for $|k|\leq r$, $\chi \in S^2$, $|y|\leq 
		{L^\pm}^{-1}$}. \label{eq:exp3}
	\end{align}
\end{assumption}

\begin{assumption}[Asymptotically Maxwellian behavior] \label{Ass:Maxwellian}
	Let $f$ satisfy the Assumptions \ref{ass:regdecay}-\ref{Ass:onepart1}.  
	Let $L^\pm=L^\pm(k,\chi)$ and $\Psi^\pm$ be as in Notation \ref{not:psidef}. We say $f$ behaves 
	asymptotically like a Gaussian if it satisfies \eqref{eq:exp1}-\eqref{eq:exp2b} and the following for some $r,C>0$:
	\begin{align}	
	|\nabla^6_{k,\chi,y} \left(\frac{F(\chi,\Psi^\pm)}{|k| \partial_u F(\chi,\Psi^\pm)}\right)| &\leq C, \quad \text{ for $|k|\leq r$, $\chi \in S^2$, $|y|\leq 
	{L^\pm}^{-1}$} \label{eq:max3}.
	\end{align}
\end{assumption}

\begin{remark}
	For example, the Assumptions \ref{Ass:Exponential} and \ref{Ass:Maxwellian} are satisfied by probability densities of the form:
	\begin{align}
		f(v) \sim \left(1+\frac{\Phi(v)}{(2+|v|^2)^{\alpha}}\right)e^{-(1+|v|^2)^\frac{\gamma}{2}}.
	\end{align}
	Here $\gamma=1$ if $f$ satisfies Assumption \ref{Ass:Exponential}, $\gamma=2$ if $f$ satisfies Assumption \ref{Ass:Maxwellian},
	$\alpha>0$ and $\Phi\in C^\infty_b$ is smooth with bounded derivatives and $|\Phi|\leq 1$. Note that this includes anisotropic velocity distributions.
\end{remark}

\subsection{Results of the paper} 

The first result of this paper is the well-posedness of the steady state equation \eqref{Bogolyubov}.
We prove that the solutions formally obtained by Oberman and Williams \cite{oberman_theory_1983}
by means of the method introduced by Lenard in \cite{lenard_bogoliubovs_1960} are indeed well-defined
solutions to the equation in the sense of Definition \ref{def:weaksol}.

\begin{theorem}[Bogolyubov correlations] \label{thm:wellpos}
	Let $f$ satisfy the Assumptions \ref{ass:regdecay} and \ref{Ass:onepart1} and
	$\phi$	be either the Coulomb potential or a soft potential. In the Coulomb case, assume further
	that $f$ satisfies Assumption \ref{Ass:Exponential} or \ref{Ass:Maxwellian}.
	Then there exists a weak solution $g_B$ to the Bogolyubov equation in the sense
	of Definition \ref{def:weaksol}. 
\end{theorem}
The proof of this theorem is the content of Subsection \ref{sec:oberwill}. 

After making precise the well-posedness of the equation, we study screening
properties of the Bogolyubov correlations. The following theorem describes the decay of the solutions of the Bogolyubov equation \eqref{Bogolyubov}. Note that the equation is written taking as unit of length the characteristic length $\ell$ of the potential in the case $\phi=\phi_S$ soft
or the Debye length $L_D$ \eqref{def:Debye} for the Coulomb potential. Therefore, the following estimate proves that the
characteristic range of interaction is given by $\ell$ or $L_D$ respectively. Furthermore, we find that the decay rate of the Bogolyubov correlations differs from the decay
rate of the potential.
 
\begin{theorem}[Screening estimate for the Bogolyubov correlations] \label{thm:lengthdecay}
	Let $f$ be a function that satisfies the  Assumptions \ref{ass:regdecay}-\ref{Ass:onepart1} and $\phi$ 
	be either Coulomb potential or a soft potential. We recall the definition of the impact parameter $b$  and
	the distance to collision $d$, as well as $d_-$ (cf. \eqref{def:Proj}).
	Then for $x\in \Reals^3$, and $v_1$,$v_2\in K$ varying on a compact set $K\subset \Reals^3$ the following estimate holds:
	\begin{align} \label{est:gB}
	|g_B(x,v_1,v_2)| \leq \frac{C(K,\delta)}{|v_r|} \frac{1}{|b|+d_-}
	\frac{1}{(1+|b|+d_-)^{\gamma-\delta}},\quad  \text{for $\delta>0$.}
	\end{align}
	If $\phi=\phi_C$, we can choose $\gamma=1$ for $f$ behaving like a Maxwellian in the sense of Assumption \ref{Ass:Maxwellian}, and
	$\gamma=0$ for $f$ satisfying Assumption \ref{Ass:Exponential}. For $\phi=\phi_S$ the statement holds for $\gamma=1$ and $C(K,\delta)$
	can be chosen independently of $K$.	
\end{theorem}
More precise estimates can be found in the Theorems \ref{thm:screen} and \ref{thm:screensoft}.

The derivation
of the Balescu-Lenard equation proposed by Bogolyubov postulates that steady states do not only exist, but are also
stable in microscopic times. More precisely, Bogolyubov's argument requires that the fluxes in $f_1$ induced
by the function $g_2$ (cf. \eqref{eq:truncBBGKY}) converge to the fluxes associated to the equilibrium correlations $g_B[f_1]$.
In the case of soft potential interaction, we prove the stability of the equilibrium correlations if
$f_1$ in \eqref{eq:truncBBGKY} is assumed to be time-independent.

\begin{theorem} \label{thm:stability}
	\label{thm:distrconv} Let $\phi$ be a soft potential and $f$ satisfy the strong stability  Assumption \ref{Ass:onepart2}. 
	Further let $\tilde{g}_{0}(\xi_1,\xi_2)=g_0(x_1-x_2,v_1,v_2)$, $g_0 \in \Sch((\Reals^3)^3)$ be 
	translation invariant and symmetric:
	\begin{align}
		\tilde{g}_0(\xi_1,\xi_2) 		&= \tilde{g}_0(\xi_2,\xi_1) 		&\quad \text{for all $\xi_1,\xi_2 \in \Reals^3 \times \Reals^3$, }		\label{eq:g0symm} \\
		\tilde{g}_0(x_1,v_1,x_2,v_2)	&= \tilde{g}_0(x_1+a,v_1,x_2+a,v_2) &\quad \text{for all $x_1,x_2,a,v_1,v_2 \in \Reals^3$. } 	\label{eq:g0trans}
	\end{align}
	Consider the function $\tilde{g}(t):=(\G(t) \tilde{g}_0)$ given by \eqref{eq:Bogolyubovprop}, which (using \eqref{eq:g0trans}) we identify with 
	\begin{align} \label{eq:greldist}
		g(t,x_{1}-x_{2},v_{1},v_{2})=	\tilde{g}(t,x_{1},v_{1},x_{2},v_{2}).
	\end{align}
	Then we have $g,\partial_t g \in  C({\mathbb{R}}^{+},\s({\mathbb{R}}^{9}))$ and $g$ solves the Bogolyubov equation \eqref{BBGKYint3g} with
	initial datum $g_0$. 
	The steady state $g_B$ given in Theorem \ref{thm:wellpos} is linearly stable, more precisely:
	\begin{align}\label{eq:distrstab}
		g(t) \longrightarrow g_B \quad \text{ in $D'(\Reals^3 \times \Reals^3 \times \Reals^3)$ as $t\rightarrow \infty$}. 
	\end{align}
	Furthermore, the associated fluxes in the space of velocities are stable, i.e. for all $v \in{\mathbb{R}}^{3}$ we have: 
	\begin{align} \label{eq:convfluxes}
	 \nabla_{v} \cdot\left( \int\nabla\phi(x) g(t,x,v,v') \ud{v'} \ud{x}\right)	\longrightarrow  \nabla_{v} \cdot\left( \int\nabla\phi(x)
	g_{B}(x,v,v') \ud{v'} \ud{x}\right) \quad \text{as $t\rightarrow \infty$}.
	\end{align}
\end{theorem}
This theorem is proved in Section \ref{sec:stability}.

\subsection{Auxiliary results}

The following lemmas provide a version of the well-known Plemelj-Sokhotski formula, which allows us to write the original function $f$ in terms of
$P^+[f]$ and $P^-[f]$ as introduced in Definition \ref{def:Pdefs}. In a more general
setting, such formulas are discussed in \cite{muskhelishvili_singular_1992}. 
\begin{lemma} \label{Lem:P}
	The operators $P^\pm$ and $P$ are bounded from $L^2$ to $L^2$. Let $f\in L^2(\Reals;\Reals)$,
	then we have $P^+[f]=\ol{P^-[f]}$. Furthermore for $f\in L^2 (\Reals;\Complex)$ there holds:
	\begin{align}\label{eq:PpmLem}	
	f= \frac{1}{2\pi i} (P^+[f] - P^-[f]).
	\end{align}
\end{lemma}
\begin{proof}
	By a classical result, $P^\pm$ are Fourier multiplication operators with
	symbols $\pm 2 \pi i \cf_{\xi>0}$. The same holds
	for $P$ with multiplier $ i \pi \sign{\xi}$. 
	Combining this with Plancherel's theorem, we find that the operators are bounded on $L^2$
	and satisfy the identity \eqref{eq:PpmLem}. For real-valued functions $f$, the
	identity $P^+[f]=\ol{P^-[f]}$ holds, since these operators are obtained in a limit $\delta\rightarrow 0$ (cf. \eqref{def:Ppm}) and the
	identity holds for all $\delta>0$.
\end{proof}
\begin{lemma} \label{lem:wienuniq}
	Let $f\in L^2(\Reals)$, and $q^+$ be analytic on the upper half plane, $q^-$ analytic
	on the lower half plane and decaying: $|q^\pm(z)|
	\rightarrow 0$, $|z| \rightarrow \infty$. Assume that $\lim_{\delta
		\rightarrow 0^+} q^\pm(\cdot\pm i\delta)$ exists in $L^2(\Reals)$ and:
	\begin{align} \label{eq:qmatching}
		\lim_{\delta \rightarrow 0^+} \frac{1}{2\pi i}
		\left(q^+(\cdot+i\delta)-q^-(\cdot-i\delta)\right) = f.
	\end{align}
	Then we have: $P^\pm[f] = q^\pm$.
\end{lemma}
\begin{proof}
	We consider the differences $\zeta^\pm := q^\pm -P^\pm[f]$. The functions are
	analytic in the upper, respectively the lower half-plane and decay as
	$|z|\rightarrow \infty$, $|\Im(z)|\geq 1$. We claim the function $\zeta$, given by $\zeta^+$ on the 
	upper half-plane and $\zeta^-$ on the lower half-plane, is an entire function. To see this, fix $z_0\in \Complex$ arbitrary and consider $Z(z):=
	\int_{\gamma{[z_0,z]}} \zeta(z') d\gamma(z')$, where $\gamma{[z_0,z]}$ is an arbitrary curve
	connecting $z_0$ and $z$. 
	Then $Z$ is an analytic function above and below and is continuous at 
	the real line by \eqref{eq:PpmLem} and \eqref{eq:qmatching}, hence an entire function. Using $Z'=\zeta$, we infer that
	$\zeta$ is an entire function as well. Outside the strip with $|\Im(z)|\leq
	1$, $\zeta$ is bounded and decays for $|z|\rightarrow \infty$. On the strip,
	we use the $L^2$ convergence of $P^\pm[f]$ and $q^\pm$ together with the mean
	value property of $\Re(\zeta), \Im(\zeta)$ to obtain:
	\begin{align*}
	|\zeta(z)| 	&\leq C \int_{B_1(z)} |\zeta(z')| \ud{z'} \leq C \left(\|f\|_{L^2} + \sup_{|r|<2} \|q^\pm(\cdot \pm ir)
			\|_{L^2(\Reals)}\right) \leq C.
	\end{align*}
	So $\zeta $ is a bounded entire function, hence constant. By $ \lim
	_{R\rightarrow \infty} \zeta(iR) = 0$ we get $\zeta\equiv 0$ as claimed.
\end{proof}	

We make Assumption \ref{Ass:onepart1} to ensure that the dielectric function $\eps$ does not vanish.
In many arguments later we will make use of quantitative lower bounds on $|\eps|$, one of which is
provided by the following lemma.

\begin{lemma}[Estimate on the degeneracy of $\eps$] \label{lem:dielectric}
	Let $f$ satisfy  the Assumptions \ref{ass:regdecay}-\ref{Ass:onepart1}. If $\phi=\phi_S$ is a soft potential, there exists $c_1>0$ such that for all $k\in \Reals^3$
	and $v\in \Reals^3$ we have:
	\begin{align}
	|\eps(k,-k\cdot v)| \geq c_1>0.\label{est:epsSoft}	
	\end{align} 
	If $\phi=\phi_C$ is the Coulomb potential, for any $K\subset \Reals^3$ compact and $\delta>0$ we have:
	\begin{align}
	|\eps(k,-k\cdot v)| &\geq c_1(K)>0, 	\quad &\text{for all $0\neq k\in \Reals^3$, $v\in K$} \label{est:epsK}\\
	|\eps(k,-k\cdot v)| &\geq c_2(\delta)>0, \quad &\text{for all $|k|\geq \delta$, $v\in \Reals^3$.} \label{est:epseps}
	\end{align}
\end{lemma}
\begin{proof}
	Let $\phi=\phi_C$ be the Coulomb potential. Then we have: 
	\begin{align}\label{eq:epsmodul}
	|\eps(k,-k\cdot v)| &= \left| 1 - \frac{1}{|k|^2} P^-[\partial_u F(\omega,\cdot)](\omega \cdot v)\right|. 
	\end{align}
	Since $|P^-[\partial_u F(\omega,\cdot)]|$ is bounded,  $|\eps(k-k\cdot v)|$ attains its minimum
	on $(k,v)\in \left(\Reals^3 \setminus B_\delta(0)\right)  \times \Reals^3$ for any $\delta>0$. This minimum
	is nonzero by \eqref{fstable1}, so \eqref{est:epseps} holds. 
	
	On the other hand, since $P^-[\partial_u F] \neq 0$ (cf. \eqref{fstable1}), the mapping  $v \mapsto \inf_{k\in \Reals^3}|\eps(k,-k\cdot v)|$ is continuous, so
	\eqref{est:epsK} holds on compact sets $K$.
	
	The estimate \eqref{est:epsSoft} for soft potentials is immediate.
\end{proof}
\begin{remark} \label{rem:couldiel}
	In the Coulomb case, the estimates \eqref{est:epsK}-\eqref{est:epseps} cannot be improved, since it is known (cf. \cite{penrose_electrostatic_1960}) that:
	\begin{align*}
		\inf_{k\in \Reals^3, v\in \Reals^3} |\eps(k,-k\cdot v)|=0.
	\end{align*}  
\end{remark}

\begin{lemma}[Asymptotics of $\alpha(\chi,u)$] 
	Let $f$ satisfy the Assumptions \ref{ass:regdecay}-\ref{Ass:onepart1}.	
	We recall the function $\alpha$ introduced in \eqref{def:alpha}. There exist
	constants $C,R>0$ such that for $|u|\geq R$:
	\begin{align}
	|\partial_u^j \alpha(\chi,u)- \frac{(-1)^{j}(j+1)!}{u^{j+2}}|&\leq \frac{C}{u^{j+3}} \quad \text{for $j\in \Naturals_0$, $j \leq 6$,} \label{est:alphabounds} \\
	|\partial_\chi^\ell \partial_u^j \alpha(\chi,u)|&\leq \frac{C}{u^{j+3}} \quad \text{for $j\in \Naturals_0$, $\ell\in \Naturals$, $j+\ell\leq 6$  } \label{est:alphabounds2}.
	\end{align}
\end{lemma}
\begin{proof}
	The derivative $\partial_u^j$ can be taken inside the operator $P$:
	\begin{align}
	\partial^j_u \alpha(\chi,u) = P[\partial^{j+1}_u F(\chi,\cdot)](u).		
	\end{align}
	Using	that $P$ is a Fourier multiplication operator with multiplier $i \pi  \sign(\xi)$ we write:
	\begin{align*}
	\widehat{\partial^j_u \alpha(\chi,\cdot)}(\xi) = i \pi \sign(\xi) \F(\partial^{j+1}_u F(\chi,\cdot))(\xi). 
	\end{align*}
	Now we perform the Fourier inversion integral and integrate by parts:
	\begin{align*}
	\partial^j_u \alpha(\chi,u)&= -\int_{-\infty}^0 (\pi/2)^\frac12  i e^{i\xi\cdot u}  \F(\partial^{j+1}_u F(\chi,\cdot))(\xi) \ud{\xi}+
	\int_0^\infty (\pi/2)^\frac12  i e^{i\xi\cdot u}  \F(\partial^{j+1}_u F(\chi,\cdot))(\xi) \ud{\xi} \\
	&= (\pi/2)^\frac12 \int_{-\infty}^0   \frac{e^{i\xi\cdot u}}{u} \partial_\xi  \F(\partial^{j+1}_u F(\chi,\cdot))(\xi) \ud{\xi} + (\pi/2)^\frac12   \frac{1}{u}   \F(\partial^{j+1}_u F(\chi,\cdot))(0) \\
	-&(\pi/2)^\frac12 \int_{0}^\infty   \frac{e^{i\xi\cdot u}}{u} \partial_\xi  \F(\partial^{j+1}_u F(\chi,\cdot))(\xi) \ud{\xi} + (\pi/2)^\frac12   \frac{1}{u}   \F(\partial^{j+1}_u F(\chi,\cdot))(0)  .
	\end{align*}
	Since $\partial^{j+1}_u F$ is a derivative, we have $ \F(\partial^{j+1}_u F(\chi,\cdot))(\xi)=0$. Iterating the argument we find:
	\begin{align}
	\partial^j_u \alpha(\chi,u) 	=  &    -\frac{(2\pi)^\frac12 i}{(-iu)^{j+2}} \partial^{j+1}_{\xi}\F(\partial^{j+1}_u F(\chi,\cdot))(0)  -(2\pi)^\frac12     -\frac{(2\pi)^\frac12 i}{(-iu)^{j+3}} \partial^{j+2}_{\xi}\F(\partial^{j+1}_u F(\chi,\cdot))(0)\label{eq:Pmleading} \\
	&+\int_{0}^\infty   \frac{e^{i\xi\cdot u}i}{(-iu)^{j+3}} \partial^{j+3}_\xi\F(\partial^{j+1}_u F(\chi,\cdot))(\xi) \ud{\xi} \notag 	-\int_{-\infty}^0   \frac{e^{i\xi\cdot u}i}{(-iu)^{j+3}} \partial^{j+3}_\xi\F(\partial^{j+1}_u F(\chi,\cdot))(\xi) \ud{\xi}.											 	
	\end{align}
	The leading order term is explicit by \eqref{Ass:fint}: 
	\begin{align} \label{eq:alphaleading}
	\partial^{j+1}_{\xi}\F(\partial^{j+1}_u F(\chi,\cdot))(0) = \frac{i^{j+1} (j+1)!}{(2\pi)^\frac12}.
	\end{align}
	Combining \eqref{eq:Pmleading}, \eqref{eq:alphaleading}  gives \eqref{est:alphabounds}. The derivative of  \eqref{eq:alphaleading} in $\chi$ vanishes, so we obtain \eqref{est:alphabounds2}.
\end{proof}
The implicit function theorem gives the following Lemma on the function $u_0$ defined in Notation \ref{not:psidef}.
\begin{lemma} \label{lem:u0lemma}
	Let $f$ satisfy the Assumptions \ref{ass:regdecay}-\ref{Ass:onepart1}. Using \eqref{est:alphabounds}, for $|k|\leq r$ , $r>0$ small enough there are  unique $u_0^\pm(k,\chi)$ such that \eqref{def:u0} holds,
	and we have the estimates:
	\begin{align}
	|\partial^j u_0^\pm(k,\chi)|&\leq \frac{C}{|k|^{j+1}} \quad  \text{for $j\in \Naturals_0$, $j \leq 6$,}  \label{est:u0bounds} \\
	|\partial_\chi^\ell \partial^j u_0^\pm(k,\chi)|&\leq \frac{C}{|k|^{j}} \quad \text{for $j\in \Naturals_0$, $\ell\in \Naturals$, $j+\ell\leq 6$  } \label{est:u0bounds2}.
	\end{align}
\end{lemma}

We can represent the solution to the Bogolyubov equation \eqref{Bogolyubov} explicitly in Fourier variables. 
The decay properties of the solution are encoded in the singularity of their Fourier transform at the origin, which
motivates to make the following definition.
\begin{definition} \label{def:strong}
	Let $0 < \kappa \leq 1$  and $f: \Reals^n \setminus \{0\}
	\rightarrow \Reals$. Define the functional $[f]_{\kappa}$ by:
	\begin{align*}
	[f]_{\kappa}(x):= \sup_{\substack{0<|h|\leq 1 \\ x+h \neq 0}}
	\frac{| f(x+h)-  f(x)|}{|h|^{\kappa}}.
	\end{align*}
\end{definition} 
The following lemma gives sharp decay estimates for functions that have an isolated singularity
in Fourier variables. 
\begin{lemma} \label{lemholder}
	Let $l\in \Naturals$, $f: \Reals^n   \setminus \{0\} \rightarrow
	\Reals$ be $\ell$ times continuously differentiable with $|\nabla^j f|\in L^1$ for $0\leq j\leq
	\ell$.
	Further let $0<\kappa \leq 1$ and $[\nabla^\ell f]_{\kappa} \in L^1$. Then the Fourier transform $\hat{f}$ decays like:
	\begin{align} \label{est:singFourier}
	|\hat{f}(x)| \leq \frac{C}{1+|x|^{\ell+\kappa}}.
	\end{align}
\end{lemma}
\begin{proof}
	Since $f \in L^1$ we know $\hat{f} \in L^\infty$ with $\|\hat{f}\|_{L^\infty}\leq C
	\|f\|_{L^1}$. For the additional decay we inspect the transformation formula
	directly. We distinguish the cases $\ell$ even and $\ell$ odd. For $\ell=2m$ even,
	we use 
	\begin{align} \label{eq:Fourderiv}
		e^{-i \pi k x}= \frac{1}{(\pi|x|)^{2m}} \Delta^m (e^{- i\pi x k}).
	\end{align}
	Further we use that $f$ is in $f \in W^{l,1}(\Reals^n)$ to compute
	\begin{align} \label{eq:Fourierest}
	\hat{f}(\pi x) 	&= \frac{1}{(2\pi)^\frac{n}{2}}\int f(k) e^{-i\pi xk} \ud{k} 	= \frac{1}{(\pi|k|)^{2m}} \frac{1}{(2\pi)^\frac{n}{2}} \int \Delta^m f(k) e^{-i \pi xk} \ud{k}. 	 
	\end{align}
	Now $g:= \Delta^m f$ satisfies $|g| +[g]_{\kappa} \in L^1$. Therefore we can
	estimate
	\begin{align*}
	\hat{g}(\pi x) 	&= - \frac{1}{(2\pi)^\frac{n}{2}}\int g(k) e^{-i \pi (k-\frac{x}{|x|^2}) x } \ud{k} 
	= \frac{1}{2(2\pi)^\frac{n}{2}}  \int \left(g(k) - g(k+\frac{x}{|x|^2})\right)
	e^{- \pi k x}	\ud{k} . 
	\end{align*}
	Taking absolute values and using $[g]_\kappa \in L^1$ gives
	\begin{align*}
	|\hat{g}(\pi x)| 	\leq  \frac{1}{2(2\pi)^\frac{n}{2}} \int [g]_{\kappa}(k)/|x|^\kappa \ud{k} 	\leq  \frac{C}{|x|^\kappa}.
	\end{align*}
	Inserting this into \eqref{eq:Fourierest} gives $|\hat{f}(x)| \leq \frac{C}{1+|x|^{l+\kappa}}$ as claimed.
	For $\ell=2m+1$ odd we repeat the computation, except that we now use
	$e^{-i \pi k x}= \frac{i x}{(\pi|x|)^{2m}} \cdot \nabla \Delta^m (e^{- i\pi x k})$ instead of \eqref{eq:Fourderiv}.
\end{proof}
As a corollary we obtain bounds for the (inverse) Fourier transform of functions
that depend on the modulus $\omega = \frac{k}{|k|}$.
\begin{lemma} \label{cordecay}
	Let $\ell \in \Naturals$, $\Phi(k,\chi)
	\in C^{n+\ell}_c( B_1(0) \times S^{n-1} )$.
	Then the Fourier transform of the mapping 
	$T(k)= |k|^\ell\Phi(k,\frac{k}{|k|})$ on $\Reals^n$ decays like:
	\begin{align*}
	|\hat{T}(x)| \leq	\frac{C(\delta)}{1+|x|^{n+\ell-\delta}}, \quad \text{for $\delta>0$ arbitrary.}
	\end{align*}
\end{lemma} 
\begin{proof}
	Follows by applying Lemma \ref{lemholder} to $T$.
	Differentiating the function we obtain the estimates:
	\begin{align*}
	[\nabla^{n+\ell-1} T]_{1-\delta}(k) \leq
	\frac{ C(\delta)|k|^\ell\|\Phi\|_{C^{n+\ell}}}{|k|^{\ell+n-\delta}}, \quad |\nabla^j T(k)| &\leq \frac{C |k|^\ell \|\Phi\|_{C^{n+\ell}}}{|k|^j}  \quad {0\leq j \leq n+\ell-1}.
	\end{align*}
	Since $T$ is compactly supported in the unit ball, we can apply  Lemma \ref{lemholder} and obtain the claim.
\end{proof}

\subsection{The Oberman-Williams-Lenard solution} \label{sec:oberwill}

The Fourier representation formula for the Bogolyubov correlations,
more precisely a Fourier representation $\hat{g}_B$ of the solution to \eqref{Bogolyubov} has been obtained by
Oberman and Williams in \cite{oberman_theory_1983}, following the complex-variable approach
by Lenard in \cite{lenard_bogoliubovs_1960}. We will briefly restate their result
in the mathematically rigorous framework of this work. We will define a function $g_B$ via its Fourier transform $\hat{g}_B$.
In order to complete the proof that $g_B$ is a solution of the Bogolyubov equation
in the sense of Definition~\ref{def:weaksol}, we need to show that $g_B$ is in $W$ and satisfies
the Bogolyubov condition \eqref{bogcond}. This is the content of Section \ref{sec:Screenest}, in particular of the Theorems \ref{thm:screen}, \ref{thm:screensoft}.  
\begin{notation} \label{Not:defAB}
	We introduce functions $A^\pm,B^\pm$, derived from $\eps$ and $F$ (cf. \eqref{def:dielectric},\eqref{def:F}):
	\begin{align}
	A^\pm(k,u)			&:= (1-B^\pm)P^\pm[\frac{F(k,\cdot)}{|\eps(k,-|k|\cdot)|^2}](u) 	\label{def:Apm} \\
	B^\pm(k,u)			&:= \hat{\phi}(k) P^\pm[\partial_u F(k,\cdot)](u)					\label{def:Bpm}. 
	\end{align}	
\end{notation}

\begin{definition} \label{def:Bogcorr}
	For $v_1,v_2 \in \Reals^3$, consider the Schwartz distribution $\hat{g}_B(\cdot,v_1,v_2) \in \Sch'(\Reals^3)$ given
	by the following linear functional $(\varphi, \hat{g}_B(v_1,v_2))_{\Sch,\Sch'}$ on $\Sch(\Reals^3)$ ($\omega$ as defined in \eqref{eq:vecshort}):
	\begin{align}
	(\varphi, \hat{g}_B(v_1,v_2))	&= \int \frac{\varphi(k)\hat{\phi}(k) \omega \left((\nabla_{v_1}-\nabla_{v_2} )(f  f) + \nabla f(v_1) \ol{\hat{h}}_B(k,v_2) - \nabla f(v_2)		\hat{h}_B(k,v_1)\right)}{\omega(v_1-v_2)-i0} \ud{k} \label{eq:Bogcorr}.
	\end{align}
	Here $-i0$ represents taking the limit $\delta \rightarrow 0^+$ with $-i\delta$ in \eqref{eq:Bogcorr}, and $\hat{h}_B$ is given by the formula:
	\begin{align}
	\hat{h}_B(k,v)		&:= f(v) \frac{(1-\eps(k,-kv))}{\eps(k,-kv)}- \hat{\phi}(k)
	\frac{A^-(k,\omega v)}{\eps(k,-kv)} (\omega \nabla f(v)) \label{eq:hFourier}.
	\end{align}	
	Then we  will call $g_B(\cdot,v_1,v_2)\in \Sch'(\Reals^3)= \F^{-1}\left(\hat{g}_B(\cdot,v_1,v_2)\right)$ the Bogolyubov
	correlation associated to $f$.
\end{definition}

The strategy for solving \eqref{Bogolyubov} is solving integrated versions of
the equation first. To fix ideas, let $g$ be a solution and consider the functions $h(x,v)$, $H(k,u)$ defined by
\begin{align*}
h(x,v_1)		&= \int_{\Reals^3} g(x,v_1,v_2) \ud{v_2} \\
\hat{H}(k,u) 			&= \int_{\Reals^3} \hat{h}(k,v) \delta(u-\frac{k v}{|k|})
\ud{v}.
\end{align*}
The key observation is that $g$, $h$ and $H$ solve the equations (as before: $\zeta(1)=2$, $\zeta(2)=1$)
\begin{align}
(v_1-v_2) \partial_x g  	&=  \sum_{j=1}^2 \nabla f(v_j) \int \nabla \phi((-1)^{j+1}x+y)
h(y,v_{\zeta(j)}) \ud{y}  + (\nabla_{v_1}-\nabla_{v_2} f)(f  f) \nabla
\phi(x) \label{eq:g} \\
\hat{h}(k,v) 				&=  \int_{\Reals^3} \frac{-\omega \hat{\phi}(k)
	((\nabla_{v_1}-\nabla_{v_2} f)(f  f)+ \nabla
	f(v_1)\ol{\hat{h}}(k,v_2)- \nabla f(v_2) \hat{h}(k,v_1)
	)}{\omega(v_1-v_2)-i0} \ud{v_2}\label{eq:h}\\
\hat{H}(k,u) 				&= 	- \hat{\phi}(k) \left(\partial_u F P^-[F] - P^-[\partial_u
F]F+  \partial_u F P^-[\ol{\hat{H}}]-  P^-[\partial_u F] \hat{H}\right) \label{eq:H}.
\end{align}
Note that the equation for $H$ is closed. This suggests to solve the equations
\eqref{eq:g}-\eqref{eq:H} in reverse order: Once we have found the solution $\hat{H}$
to \eqref{eq:H}, we can use \eqref{eq:h} to compute $\hat{h}$ and then compute $\hat{g}$ using \eqref{eq:g}.
Following this reasoning, we show the existence of a solution to \eqref{eq:H} in the first step of our rigorous analysis.
\begin{lemma} \label{lem:H} Let $f$ satisfy the Assumptions \ref{ass:regdecay}-\ref{Ass:onepart1}.
	We recall the definitions of $F$ in \eqref{def:F} and $A^\pm$ in \eqref{def:Apm}. 
	The function $\hat{H}_B: \Reals^3 \times \Reals\rightarrow \Reals$ given by
	\begin{align}
	\hat{H}_B(k,u) := \frac{1}{2\pi i}(A^{+} - A^-) - F(k,u) \label{Hdef}
	\end{align}
	is measurable in $\Reals^3 \times \Reals$ and satisfies $\hat{H}_B(k,\cdot) \in L^2$ a.e. in 
	$k\in \Reals^3$. Further, for a.e. $k\in \Reals^3$ it solves the equation:
	\begin{align} \label{eq:H1}
	\hat{H}_B(k,u) = - \hat{\phi}(k) \left(\partial_u F P^-[F] - P^-[\partial_u
	F]F+ \partial_u F P^-[\ol{\hat{H}_B}]- P^-[\partial_u F] \hat{H}_B\right).
	\end{align}
\end{lemma}
\begin{proof}
	As a pointwise a.e. limit of measurable functions, $\hat{H}_B$ is measurable again.
	By Lemma \ref{Lem:P} we know that $A^+ = \ol{A^{-}}$, so $\hat{H}_B$ is real-valued. By 
	\eqref{est:epseps} $|\eps|$ is bounded below,
	so $\frac{F}{|\eps|}$ is $L^2$. We can rewrite $A^-$ using $\eps$ (as in cf. \eqref{def:eps}):
	\begin{align} \label{eq:Amterm}
	A^-(k,\cdot) = \eps(k,-|k|\cdot) \int_\Reals
	\frac{F(\omega,u')}{|\eps(k,-|k|u')|^2(u'-\cdot+i0)} \ud{u'},
	\end{align}
	and find this function is in $L^2$, since $P^\pm$ are bounded on $L^2$. It remains to show that $\hat{H}_B$ satisfies the equation.
	Since $\hat{H}_B$ is real-valued, equation \eqref{eq:H1} is equivalent to
	\begin{align*} 
	\hat{H}_B+F &= F -\hat{\phi}(k)	\left(\partial_u F P^-[F+\hat{H}_B] - (F+\hat{H}_B) P^-[\partial_u F] \right).
	\end{align*} 	
	Using that	$|1-\hat{\phi}(k) P^+[\partial_u F]|=|\eps|$ is non-zero, Lemma \ref{Lem:P} shows that the equation
	is equivalent to: 
	\begin{align} \label{eq:Hwien} 
	\frac{P^+[\hat{H}_B+F]}{1- \hat{\phi}(k) P^+[\partial_u F]} -
	\frac{P^-[\hat{H}_B+F]}{1-\hat{\phi}(k) P^-[\partial_u F]} = \frac{2\pi i
		F(u)}{(1-\hat{\phi}(k) P^+[\partial_u F])(1-\hat{\phi}(k) P^-[\partial_u F])}.
	\end{align}
	So it remains to check \eqref{eq:Hwien} is satisfied for $\hat{H}_B$ as defined in
	\eqref{Hdef} above. The equation is satisfied, if we can show that 
	\begin{align} \label{eq:Aplus}
	P^\pm [\hat{H}_B] = A^\pm -	P^\pm[F].	
	\end{align} 
	By the definition \eqref{Hdef} of $\hat{H}_B$, this is the case if for $A^\pm$
	as in \eqref{def:Apm} we have:
	\begin{align} \label{eq:Apmident}
	A^\pm = P^\pm[\frac{1}{2\pi i}(A^+-A^-)].
	\end{align}
	This however follows from the uniqueness proved in Lemma \ref{lem:wienuniq}.
\end{proof}
\begin{lemma} \label{lem:h}
	Let $f$ satisfy the Assumptions \ref{ass:regdecay}-\ref{Ass:onepart1} and
	consider the function $\hat{h}_B$ defined by the Fourier representation
	\eqref{eq:hFourier}. 
	Then $\hat{h}_B$ is a measurable function in $\Reals^3 \times \Reals^3$ and for $k\neq 0$ it satisfies:
	\begin{align} \label{eq:hdecay}
	|\hat{h}_B(k,v)| \leq C(k) e^{-|v|}
	\end{align}
	Furthermore, for $k\neq 0$ the function $\hat{h}_B(k,\cdot)$, $k\neq 0$ solves the equation:
	\begin{align}\label{eq:heq}
	\hat{h}_B(k,v) &=   \int_{\Reals^3} \frac{\omega \hat{\phi}(k)
		((\nabla_{v_1}-\nabla_{v_2} f)(f  f)+ \nabla
		f(v_1)\ol{\hat{h}}_B(k,v_2)- \nabla f(v_2) \hat{h}_B(k,v_1)
		)}{\omega(v_1-v_2)-i0} 
	\ud{v_2}.
	\end{align}
\end{lemma}
\begin{proof}
	Measurability and decay of $\hat{h}_B$ follow  from the regularity and
	decay properties of $f$. It
	remains to show $\hat{h}_B(k,\cdot)$ solves \eqref{eq:heq}. To this
	end, we first show $H_*(k,\cdot) := \int_{\Reals^3}\hat{h}_B(k,v)
	\delta(\cdot -\omega v) \ud{v} $ coincides with the function $\hat{H}_B(k,\cdot)$
	(cf. \eqref{Hdef}). This
	can be seen by integrating 	\eqref{eq:hFourier}:
	\begin{align*}
	H_*(k,u) = F(k,u) \frac{1-\eps(k,-|k|u)}{\eps(k,-|k|u)} -
	\frac{A^-(k,u)}{ \eps(k,-|k|u)} \frac{1}{2 \pi i}(B^+-B^-).
	\end{align*}	
	Since $\eps(k,-|k|)= 1- B^-(k,u)$, the claim $\hat{H}_B=H_*$ is equivalent to
	verifying
	\begin{align} \label{Hver}
	\hat{H}=\frac{1}{2\pi i}(P^+[\hat{H}]-P^-[\hat{H}])= \frac{F B^-}{1-B^-} -
	\frac{A^-}{1-B^-}\frac{1}{2\pi i}(B^+ - B^-).
	\end{align}
	We add $F$ on both sides and use \eqref{eq:Aplus} to see this
	is equivalent to
	\begin{align*}
	\frac{1}{2\pi i}(A^+-A^-)= \frac{F B^-}{1-B^-} -
	\frac{A^-}{1-B^-}\frac{1}{2\pi i}(B^+ - B^-) + F.
	\end{align*}
	Rearranging terms, the claim can be rewritten as:
	\begin{align*}
	\frac{1}{2\pi i}(A^+(1-B^-)-A^-(1+B^+))=  F,
	\end{align*}
	which is equivalent to \eqref{eq:Hwien}. Hence we have verified
	\eqref{Hver} and proven $H_*=\hat{H}_B$.
	Using this we can prove $\hat{h}_B$ as defined above solves \eqref{eq:heq}. To
	this end, we integrate in $v_2$ and bring the last summand in
	\eqref{eq:heq} to the left-hand side, when the equation reads:
	\begin{align*}
	\eps(k,-kv) \hat{h}_B(k,v_1) 	&=  \int_{\Reals^3}
	\frac{\hat{\phi}(k) \omega }{\omega \cdot (v_1-v_2)-i0} \left((\nabla_{v_1}-\nabla_{v_2} f)(f  f)(v_1,v_2)+ \nabla
	f(v_1)\ol{\hat{h}_B}(k,v_2)\right) \ud{v_2} \\
	&= - \hat{\phi}(k)  \left(\omega \nabla f(v_1) P^-[F+\hat{H}_B]- P^-[F]
	f(v)\right) .
	\end{align*}
	Replacing $P^-[F+\hat{H}_B]=A^-$ by means of \eqref{eq:Aplus}, we have shown the claim to
	be equivalent to \eqref{eq:hFourier}, the definition of $\hat{h}_B$.	
\end{proof}
Now it is straightforward to check that $g_B$ defined in Definition \ref{def:Bogcorr} is a weak solution of the Bogolyubov equation,
assuming that $g_B$  has marginal $\int \hat{g}_B(x,v_1,v_2)= h_B(x,v_1)$ and satisfies the Bogolyubov boundary
condition \eqref{bogcond}. These conditions will be proved in the Theorems \ref{thm:screen}, \ref{thm:screensoft}, whose
proof does not depend on the results in this section. 
\begin{theorem} \label{thm:existence}
		Let $f$ satisfy the  Assumptions \ref{ass:regdecay} and \ref{Ass:onepart1} and
		$\phi$	be either the Coulomb potential or a soft potential. In the Coulomb case, assume further
		that $f$ satisfies Assumption \ref{Ass:Exponential} or \ref{Ass:Maxwellian}.
	If $g_B$ defined by \eqref{def:Bogcorr} satisfies $\int \hat{g}_B(x,v_1,v_2)= h_B(x,v_1)$, and the Bogolyubov boundary
	condition \eqref{bogcond}, then $g_B$ is a weak solution to the Bogolyubov equation.
\end{theorem}
\begin{proof} 
	Since $g\in W$ by assumption, the equation \eqref{eq:bogolyubov} holds weakly if the Fourier-transformed equation 
	\begin{align}
		(v_1-v_2)ik \hat{g}_B - i k \hat{\phi} \nabla f(v_1) \hat{h}_B(k,v_2) + i\hat{\phi} \nabla f(v_1) \ol{\hat{h}}_B(k,v_2) = ik (\nabla_{v_1}-\nabla_{v_2})(f f)\hat{\phi} ,
	\end{align}
	 holds in the sense of distributions. This is true by the definition of $g_B$ (cf. \eqref{def:Bogcorr}).	
\end{proof}

\section{Characteristic length scale of the equilibrium correlations} \label{sec:Screenest}

In this section, we estimate the Bogolyubov correlations $g_B$, and give sufficient conditions for the onset of a characteristic length scale. In the Coulomb case,
we observe the onset of a characteristic length scale for one-particle functions $f$ that behave like Maxwellians for large velocities, and the characteristic length is given by the Debye length $L_D$ (cf. \eqref{def:Debye}). In the soft potential case, the Bogolyubov correlations always have a characteristic length scale, which coincides with the length scale of the potential. For both types of potentials, we derive the rate of decay. This will provide the assumptions on $h_B$, $g_B$ made in Theorem \ref{thm:existence}, and hence complete
the proof of Theorem \ref{thm:wellpos}.

To this end, for $v_1,v_2\in \Reals^3$ we define $\hat{\Gamma}(\cdot,v_1,v_2)\in \Sch'(\Reals^3)$ by:
\begin{align} \label{def:gamma}
	\hat{\Gamma}(k,v_1,v_2)&:= \hat{\phi}(k)k \left((\nabla_{v_1}-\nabla_{v_2} f)(f  f)+\nabla f(v_1) \ol{\hat{h}}_B(k,v_2)-\nabla f(v_2) \hat{h}_B(k,v_1)\right).
\end{align}
This allows us to get a representation of $g_B$ (cf. \eqref{eq:Bogcorr}) of the form:
\begin{align}
	\hat{g}_B(k,v_1,v_2)=\frac{1}{k(v_1-v_2)-i0} \hat{\Gamma}(k,v_1,v_2).	
\end{align}
Using the notation introduced in \eqref{eq:vecshort}, this yields the identity:
\begin{align} \label{eq:gBconvrep}
	g_B(x,v_1,v_2) &= \frac{2\pi i}{|v_r|} \Gamma(x,v_1,v_2) *_x \left( \cf_{(0,\infty)}(x \cdot \vr) \cdot \Haus^1 \llcorner \operatorname{span} \{\vr\}\right).	
\end{align}
Here we have used the one-dimensional Fourier transform $\F^{-1}(\frac{1}{\cdot-i0})=(2\pi)^\frac12 i \cf_{(0,\infty)}(\cdot)$, and the notation
$\Haus^1 \llcorner Y$ for the one-dimensional Hausdorff-measure supported on a line $Y$. 
The properties of the equilibrium correlations $g_B$ can be analyzed by first characterizing the properties of $\Gamma$, and then using the convolution representation \eqref{eq:gBconvrep}. 

\subsection{Coulomb interaction}

In this paragraph, we analyze the onset of a characteristic length in the Bogolyubov correlations $g_B$ (cf. \eqref{eq:Bogcorr}) in the case of Coulomb interacting particles. Taking the Debye length $L_D$ (cf. \eqref{def:Debye}) as unit of length, the Bogolyubov equation has the form \eqref{Bogolyubov} with $\phi=\phi_C$. 
The result we will prove in this paragraph is the following.
\begin{theorem}[Screening in the Coulomb case] \label{thm:screen}
	Let $g_B$ be defined by \eqref{eq:Bogcorr}, where $f$ satisfies
	the Assumptions  \ref{ass:regdecay}-\ref{Ass:onepart1} and  $\phi=\phi_{C}$ is
	the Coulomb potential (cf. Definition \ref{def:potentials}). 
	Further let $f$ satisfy Assumption \ref{Ass:Maxwellian} (Maxwellian behavior for $|v|\rightarrow \infty$)
	or Assumption \ref{Ass:Exponential} (Exponential behavior for $|v|\rightarrow \infty$).
	Then the marginal of $g_B$ coincides
	with $h_B$:
	\begin{align} \label{eq:hismarginal}
		\int g_B(x,v_1,v_2) \ud{v_2} = h_B(x,v_1).
	\end{align} 
	We recall the definition of  $v_r$ in \eqref{eq:vecshort}, and $b,d,d_-$ in \eqref{def:impactdist}. Let
	$K\subset \Reals^3$ be compact and $\delta \in (0,1)$.
	Under Assumption \ref{Ass:Maxwellian}, $g_B$, $h_B$ satisfy the following estimates for $x \in \Reals^3$, $v_1,v_2\in K$:
	\begin{align} 
		|g_B(x,v_1,v_2)| &\leq \frac{C(K,\delta)}{|v_r|} 
		\frac{1}{(|b|+d_-)(1+|b|+d_-)^{1-\delta}} \label{est:gBMax}, \\
		|h_B(x,v_1)| 	&\leq \frac{C(K,\delta)} {|x|(1+|x|^{3-\delta})}. \label{est:hcoulM}
	\end{align}
	Under Assumption \ref{Ass:Exponential}, $g_B$, $h_B$ satisfy the following estimates for $x \in \Reals^3$, $v_1,v_2\in K$:
	\begin{align}
		|g_B(x,v_1,v_2)| 	&\leq \frac{C(K,\delta)}{|v_r|}
		\frac{(1+|b|+d_-)^{\delta}}{(|b|+d_-)} \label{est:gBExp},  \\
		|h_B(x,v)| 		&\leq \frac{C(K,\delta)} {|x|(1+|x|^{2-\delta})}. \label{est:hcoulE}
	\end{align}
\end{theorem}
Note that the result \eqref{est:gBMax} shows the onset of a characteristic length in the correlations $g_B$ if $f$ satisfies
Assumption \ref{Ass:Maxwellian}, but the estimate \eqref{est:gBExp} indicates this is not in general true for functions satisfying Assumption \ref{Ass:Exponential}.
Furthermore, the estimates \eqref{est:gBMax} and \eqref{est:gBExp} prove that $g_B$ satisfies the Bogolyubov boundary condition \eqref{bogcond}. 

For estimating the decay of the function $g_B$, we use Lemma \ref{cordecay}, i.e.
we expand the Fourier transform of $h_B$ near $k=0$ into  
\begin{align}\label{eq:happrox}
	\hat{h}_B(k,v) = |k|^r T(k,\omega,v), 
\end{align}
where $T$ is some smooth function. Note that
the representation formula for $\hat{h}_B$ \eqref{eq:hFourier} suggests
that \eqref{eq:happrox} holds with $r=-2$, in which case Lemma \ref{cordecay} gives an estimate of $|h(x,v)| \leq C/|x|$ for $|x|\rightarrow \infty$. In other words, naively
one might expect the decay of the correlations to be the same as the decay of the Coulomb potential. 
However, since $\hat{\phi}(k)$ appears also in the dielectric constant $\eps$ in the denominator, we obtain $r>-2$ in \eqref{eq:happrox}.
Computing the precise value of $r$ is subtle, since the denominator $|\eps(k,-|k|u')|^2$ in $P^-[A]$ (appearing in \eqref{eq:hFourier}) becomes singular for $|u'|\sim 1/|k|$, $k\rightarrow 0$ as observed in Remark \ref{rem:couldiel}. The following lemma allows to separate the critical region from the remainder.  
\begin{lemma} \label{lem:Tdec}
	Assume that $f$ satisfies the Assumptions \ref{ass:regdecay}-\ref{Ass:onepart1} and $\phi=\phi_C$ is
	the Coulomb potential.
	There exists $r_0>0$ and $T(k,\chi,v) \in C^6(B_{r_0}(0)\times S^2 \times \Reals^3)$ such that for $|k|\in (0,r_0)$, $\chi\in S^2$, $v\in \Reals^3$:
	\begin{equation} \label{eq:Tdecompos}
		\begin{aligned}
			\int_{\Reals} \frac{\hat{\phi}(k)(\omega \cdot \nabla f(v)) F(\omega,u')}{|1-\hat{\phi}(k)\alpha^-(\omega, u')|^2(\omega\cdot v-u'+i0)} \ud{u'}
			= D(k,\omega,v)  +|k|^2 T(k,\omega,v) .	
		\end{aligned}
	\end{equation}
	Here $D$ is given by the formula ($u_0^\pm, I$ as in Notation \ref{not:psidef})
	\begin{align} \label{def:D}
	D(k,\chi,v) = \int_{I(k,\chi)} \frac{\hat{\phi}(k)(\chi \cdot \nabla f(v)) F(\chi,u')}{|1-\hat{\phi}(k)\alpha^-(\chi, u')|^2(\chi\cdot v-u')} \ud{u'}.
	\end{align}
	Moreover,  $T$ satisfies the estimate:
	\begin{align} \label{est:T}
	\|T(\cdot,\cdot,v)\|_{C^6(B_{r_0}(0)\times S^2)} &\leq C.	
	\end{align}
\end{lemma}
\begin{proof}
	We decompose $\alpha^-$ (cf. \eqref{def:alpha}) into its real and imaginary part:
	\begin{align}
	\alpha^-(\chi,u) &= \alpha(\chi,u) -i\pi \partial_u F(\chi,u).
	\end{align}	
	By Lemma \ref{lem:u0lemma}, for $|k|\in (0,r_0)$ small enough and $\chi \in S^2$ there exist $u_0^\pm(k,\chi)$ such that 
	\eqref{def:u0} holds. 
	By the estimate \eqref{est:alphabounds}, after possibly choosing a smaller $r_0>0$, the following holds for $|k|\in (0,r_0)$ and $u \neq I(k,\chi)$:
	\begin{align} \label{est:alphabound2}
	\frac{1}{||k|^2 + \alpha^-(\chi,u)|}\leq C(1+|u|^3).
	\end{align}
	Now the claim follows by decomposing:
	\begin{align*} 
	&\int_{\Reals} \frac{\hat{\phi}(k)(\omega \cdot \nabla f(v)) F(\omega,u')}{|1-\hat{\phi}(k)\alpha^-(\omega, -|k|u')|^2(\omega\cdot v-u'+i0)} \ud{u'} \\
	=	& |k|^2\int_{\Reals \setminus I(k)} \frac{(\omega \cdot  \nabla f(v)) F(\omega,u')}{||k|^2 + \alpha^-(\omega,u')|^2(\omega v-u'+i0)} \ud{u'}
	+ \int_{I(k)} \frac{\hat{\phi}(k)(\omega \cdot \nabla f(v)) F(\omega,u')}{|1-\hat{\phi}(k)\alpha^-(\omega, u')|^2(\omega v-u')} \ud{u'}, 	
	\end{align*}
	since by \eqref{est:alphabound2} the function $T$ given by:
	\begin{align}
	T(k,\chi,v):=\int_{\Reals \setminus I(k)} \frac{(\chi \cdot \nabla f(v)) F(\chi,u')}{||k|^2 + \alpha^-(\chi,u')|^2(\chi\cdot v-u'+i0)} \ud{u'}	
	\end{align}
	satisfies the estimate \eqref{est:T}.
\end{proof}
Now we have decomposed the integral \eqref{eq:Tdecompos} into a well-behaved part $T$, and the singular integral $D$. The behavior of $D$ for large $v$ depends on the behavior of $f$ as $v\rightarrow \infty$. If $f$ behaves like a Maxwellian, we have $D(k,v)\approx |k|$ for small $k$. If $f$
behaves like an exponential, the function is of order one close to the origin.

\begin{lemma}[Expansion of $D$ at $k=0$] \label{lem:Dest}
	Let $f$ satisfy the Assumptions \ref{ass:regdecay}-\ref{Ass:onepart1}. Rewrite the function $D$ defined by \eqref{def:D} in the following form:
	\begin{align}
	D(k,\chi,v) &=  		\gamma_h(k,\chi,v)  \quad &\text{ if $f$ satisfies Assumption \ref{Ass:Exponential},} \label{exp:Dexp}\\
	D(k,\chi,v) &= |k| 	\gamma_h(k,\chi,v) \quad 	&\text{ if $f$ satisfies Assumption \ref{Ass:Maxwellian}}  \label{exp:DMax}.
	\end{align}
	We can choose $\gamma_h \in C(B_{r_0}(0) \times S^2 \times \Reals^3)$ ($r_0$ as in Lemma \ref{lem:Tdec}) such that for any $K\subset \Reals^3$ compact  
	\begin{align}
	|\nabla^j_{k,\chi} \gamma_h(k,\chi,v)| \leq C(K), \quad  \text{for $j=0,1\ldots,6$, $k\in B_{r_0}(0)$, $\chi \in S^2$ and $v\in K$.} \label{est:D}
	\end{align}
	Similarly, for $\chi \in S^2$, $k\in B_{r_0}(0)$, $v_1,v_2\in \Reals^3$ write :
	\begin{align}
	  \chi (\nabla f(v_2) D(k,\chi,v_1) -   \nabla f(v_1) D(k,\chi,v_2))  = |k| \gamma_g (k,\chi,v_1,v_2) \quad  &\text{under  Assumption \ref{Ass:Exponential},}  \label{exp:symDexp}\\
	\chi (\nabla f(v_2) D(k,\chi,v_1) -   \nabla f(v_1) D(k,\chi,v_2))  = |k|^2 \gamma_g (k,\chi,v_1,v_2) \quad  &\text{under  Assumption \ref{Ass:Maxwellian}}\label{exp:symDMax}.
	\end{align}
	In both cases, we can choose $\gamma_g \in C(B_{r_0}(0) \times S^2 \times \Reals^3 \times \Reals^3)$ such that for all $K\subset \Reals^3$ compact:
	\begin{align} \label{est:Dsym}
	|\nabla^j_{k,\chi} \gamma_g(k,\chi,v_1,v_2)| \leq  C(K), \quad  \text{for $0\leq j\leq 6$, $k\in B_{r_0}(0)$, $\chi \in S^2$ and $v_1,v_2 \in K$.}	
	\end{align}
\end{lemma}
\begin{proof}
	After changing variables with $\Psi(k,\chi,\cdot)$, $D$ reads:
	\begin{align}
	D(k,\chi,v) &= \int_{-1/L}^{1/L} \frac{|k|^2 \chi \nabla f(v) F(\chi,\Psi(y)) L(k,\chi)}{||k|^2-\alpha(\Psi(y))|^2 + |\partial_u F(\chi,\Psi(y))|^2} \frac{1}{\chi \cdot v-\Psi(y)} \ud{y} \\
	&=  \int_{-1/L}^{1/L} \frac{|k|^3}{\partial_u \alpha(\chi,\Psi)}\frac{ \chi \nabla f(v) 
		(F/\partial_u F)(\chi,\Psi) }
	{\left| \frac{|k|^2-\alpha(\Psi)}{y\partial_u F(\chi,\Psi)}\right|^2 y^2 + 1} \frac{|k|^{-1}}{\chi \cdot v-\Psi(y)} \ud{y}.
	\end{align}
	If $f$ satisfies Assumption \ref{Ass:Exponential}, then for $|k|\leq \lambda$ small enough, the functions $F/\partial_u F$, $\frac{|k|^3}{\partial_u \alpha(\chi,\Psi)} $ and
	$\frac{|k|^2-\alpha(\Psi)}{y\partial_u F(\chi,\Psi)}$ are bounded, as well as their derivatives in $k,\chi$. Furthermore,
	$|\frac{|k|^2-\alpha(\Psi)}{y\partial_u F(\chi,\Psi)}|\geq c >0$ is bounded below. 
	Additionally, we use $\psi(k,\chi, y)\in I(k,\chi)$ and $|\chi \cdot v|\leq C(K)$ to infer that the function 
	\begin{align}
	z(k,\chi,v,y) = \frac{|k|^{-1}}{\chi \cdot v-\Psi(y)}	
	\end{align}
	is bounded as well as its derivatives in $k,\chi$. Hence, under Assumption \ref{Ass:Exponential}
	the expansion \eqref{exp:Dexp} with the estimate
	\eqref{est:D} follow by differentiating through the integral. Similarly, we prove \eqref{exp:DMax} with the estimate \eqref{est:D} under Assumption	\ref{Ass:Maxwellian}.
	
	The expansions \eqref{exp:symDexp}-\eqref{exp:symDMax} with the estimate \eqref{est:Dsym} are proved analogously, using the fact that
	\begin{align}
	z_{sym}(k,\chi,v_1,v_2,y) 	= (\frac{|k|^{-2}}{\chi \cdot v_1 -\Psi(y)} - \frac{|k|^{-2}}{\chi \cdot v_2 -\Psi(y)} )
	= \frac{|k|^{-2}\chi(v_2-v_1)}{(\chi \cdot v_1 -\Psi(y))(\chi \cdot v_1 -\Psi(y))},
	\end{align}
	is a bounded function, as well as its derivatives in $k,\chi$. 
\end{proof}
We now prove an integral estimate for $\hat{h}_B(k,v)$ (cf. \eqref{eq:hFourier}).
\begin{lemma} \label{lem:hbound}
	Let $f$ satisfy the Assumptions \ref{ass:regdecay}-\ref{Ass:onepart1}, and Assumption \ref{Ass:Exponential} or \ref{Ass:Maxwellian}.
	Further let $\phi=\phi_C$ be the Coulomb potential and $h_B$ be given by \eqref{eq:hFourier}. Then there exists $C>0$ such that
	\begin{align} \label{eq:hbound}
		\int_{B_2} \left|\int_{\Reals^3} \hat{h}_B(k,v) \ud{v} \right|\ud{k}\leq C. 
	\end{align}
\end{lemma}
\begin{proof}
	We start by performing the integration in the direction orthogonal to $\omega$ using Fubini's Theorem:
	\begin{align} \label{est:hintegrated} 
			\int_{B_2} &\left|\int_{\Reals^3} \hat{h}_B(k,v) \ud{v} \right|\ud{k} =\int_{B_2} \left|\int_{\Reals} \int_{\Reals^3} \hat{h}_B (k,v) \delta(u-\omega v) \ud{v}\ud{u}\right| \ud{k} \notag	\\
		\leq &\int_{B_2} \left|\int_{\Reals} F(\omega,u) \frac{1-\eps(k,-|k|u)}{\eps(k,-|k|u)} -\hat{\phi}(k) \frac{A^-(k,u)}{\eps(k,-|k|u)} \partial_u F(k,u)\ud{u}\right| \ud{k} \notag \\
		\leq &C+ \int_{B_2} \left|\int_{\Reals}  \frac{F(\omega,u)}{\eps(k,-|k|u)}\ud{u}\right| \ud{k} + \int_{B_2} \left|\int_{\Reals}  \hat{\phi}(k) \frac{A^-(k,u)}{\eps(k,-|k|u)} \partial_u F(k,u)\ud{u}\right| \ud{k} .
	\end{align}
	Now the estimates follow similar to the proof of the last Lemma. 
	We observe that for $|k|\geq \lambda>0$ bounded away from the origin,  the integrand in the first integral in \eqref{est:hintegrated} is bounded. Further, for
	$\lambda>0$ small enough we know that $|F(u)/\eps(k,-|k|u)|\leq |F(u)/\partial_u F|$ is bounded for $|u-u_0^\pm(k,\omega)|\leq 1$. Finally, on the region
	$|k|\leq \lambda$, $|u-u_0|\geq 1$, the integral is bounded since $|\eps(k,-|k|u)|^{-1}\leq C(1+|u|^3)$.
	
	In order to bound the second integral in \eqref{est:hintegrated}, we recall the definition of $A^-$ \eqref{def:Apm} to rewrite:
	\begin{align*}
		\int_{B_2} \left|\int_{\Reals}  \hat{\phi}(k) \frac{A^-(k,u)}{\eps(k,-|k|u)} \partial_u F(k,u)\ud{u}\right| \ud{k} = \int_{B_2} \left|\int_{\Reals}  \hat{\phi}(k) P^-[\frac{F(k,\cdot)}{|\eps(k,-|k|\cdot)|^2} ](u)\partial_u F(k,u)\ud{u}\right| \ud{k}.
	\end{align*}
	Now the claim follows if we can show that $\left|\int P^-[\frac{F}{|\eps|^2} ](u)\partial_u F(k,u)\ud{u}\right|\leq C$
	is uniformly bounded, for $|k|$ sufficiently small.
	For $I(k,\omega)$ as introduced in \eqref{eq:Idef} we can estimate
	\begin{align}\label{est:hintegrated2}
		\left|\int P^-[\frac{F}{|\eps|^2} ](u)\partial_u F(k,u)\ud{u}\right| \leq 
		C +\left|\int_{I(k)} \int_{ I(k)}\frac{F(k,u') \partial_u F(u)}{|\eps(k,-|k|u')|^2(u-u'-i0)}  \ud{u'} \ud{u}\right| .
	\end{align}
	Now since $f$ satisfies Assumption \ref{Ass:Exponential} or \ref{Ass:Maxwellian}, the function
	$\frac{F(k,u') \partial_u F(u)}{|\eps(k,-|k|u')|^2}$ and its derivative in $u'$ is bounded
	for $u,u'\in I(k)$ and $|k|$ sufficiently small. Therefore, the integral \eqref{est:hintegrated2} is uniformly bounded and the claim follows.
\end{proof}
From the expansion of $D$ near $k= 0$ in Lemma \ref{lem:Dest}, we can now obtain an expansion 
of $\hat{h}_B$ and $\hat{g}_B$ near $k=0$.

\begin{lemma}[Expansion of $\hat{h}_B$ for $|k|\rightarrow  0$ and $|k|\rightarrow \infty$] \label{lem:hexpansion}
	Assume that $f$ satisfies the Assumptions \ref{ass:regdecay}-\ref{Ass:onepart1} and $\phi=\phi_C$ is
	the Coulomb potential.
	Let $\hat{h}_B$ be given by	\eqref{eq:hFourier} and $K \subset\Reals^3$ compact.
	Then there exists a function $\hat{h}_{B,0}(k,\chi,v)\in C^6(B_1(0)\times S^2 \times \Reals^3)$ such that:
	\begin{align}
		\|\hat{h}_{B,0}(\cdot,\cdot,v)\|_{C^6(B_1(0)\times S^2)}& \leq C(K), \quad &\text{for $v\in K$} \label{est:hBOest}\\
	\hat{h}_B(k,v) &= -f(v) + |k|\hat{h}_{B.0}(k,k/|k|,v), \quad 	&\text{under Assumption \ref{Ass:Maxwellian}} \label{eq:hMax} \\
	\hat{h}_B(k,v) &= -f(v) + \hat{h}_{B.0}(k,k/|k|,v), 	\quad 	&\text{under Assumption \ref{Ass:Exponential}} \label{eq:hExp}.
	\end{align}	
	Furthermore for  $|k|\geq 1$ and $\ell\in 1,\cdots,6$ we have:
	\begin{align}
	|\nabla^\ell_k \hat{h}_B(k,v)| \leq \frac{C}{1+|k|^{\ell+2}}e^{-|v|} \label{est:hinfty}. 
	\end{align}
\end{lemma}
\begin{proof}
	On the region $|k|\in (r_0,1)$, the function $\hat{h}_B(k,v)$ is smooth by \eqref{est:epseps}.
	For $|k|\in (0,r_0)$ small, we use $\hat{\phi}(k)=\frac{1}{|k|^2}$ and the decomposition \eqref{eq:Tdecompos}:
	\begin{align}
	\hat{h}_B(k,v)&= -f(v) + |k|^2 \left(\frac{ f(v) }{|k|^2 -\alpha^-(\omega\cdot v)+ i \partial_u F(\omega,\omega \cdot v)})
	- T(k,\omega,v) \right) + D(k,\omega,v).
	\end{align}
	The first two summands can be written in the forms \eqref{eq:hMax}, \eqref{eq:hExp} respectively, as can be inferred from
	from Lemma \ref{lem:Tdec} and \eqref{fstable1}. For the last summand, the claim follows from Lemma \ref{lem:Dest}.
	It remains to prove the estimate \eqref{est:hinfty}. This however follows from the lower bound \eqref{est:epseps} on $|\eps|$ for $|k|\geq 1$.
\end{proof}
\begin{proofof}[Proof of Theorem \ref{thm:screen}]
	Let  $\eta\in C^\infty_c$ be
	a cutoff function with $\eta(k)=1$ for $|k|\leq 1/2$ and $\eta(k)=0$ for $|k|\geq 1$.
	We recall the functions $\Gamma$ (cf. \eqref{def:gamma}) and $h_B$ (cf. \eqref{eq:hFourier}),
	and separate the contributions of large and small Fourier modes:
	\begin{align}
		\hat{\Gamma}(k,v_1,v_2) &= \eta(k) \hat{\Gamma} + (1-\eta)(k) \hat{\Gamma} =: \hat{\Gamma}_1+\hat{\Gamma}_2  \label{GammaDecomp}\\
		\hat{h}_B(k,v) &= \eta(k) \hat{h}_B + (1-\eta)(k) \hat{h}_B =: \hat{h}_{B,1}+\hat{h}_{B,2} \label{hDecomp}.
	\end{align}
	The function $h_{B,1}$ satisfies the estimates \eqref{est:hcoulM},\eqref{est:hcoulE}, which
	can be seen by applying Lemma \ref{lemholder} to the
	expansions \eqref{eq:hMax},\eqref{eq:hExp}. The function $h_{B,2}$ satisfies
	the estimates \eqref{est:hcoulM},\eqref{est:hcoulE} by \eqref{est:hinfty}.
	
	In order to estimate $\Gamma_1$, we again apply Lemma \ref{lemholder}. To this end,
	we insert the expansion of $\hat{h}_B$ into the definition of $\Gamma$ (cf. \eqref{def:gamma})
	to find:
	\begin{align*}
		\hat{\Gamma}(k,v_1,v_2) = k/|k|^2 (\nabla f(v_1) h_{B,0}(k,v_2) - \nabla f(v_2) h_{B,0}(k,v_1)), \quad \text{ for $|k|\leq 1$}.
	\end{align*}
	Hence for any $\delta>0$ and $R>0$, Lemma \ref{lemholder} shows that $\Gamma_1$ decays like
	\begin{align} \label{est:gamma1}
		|\Gamma_1(x,v_1,v_2)|\leq \frac{C(K,\delta)}{1+|x|^{m-\delta}}, \quad \text{for $x\in \Reals^3$, $|v_1|,|v_2|\leq  R$},
	\end{align}
	where $m=3$ if $f_1$ satisfies Assumption \ref{Ass:Maxwellian}, and $m=2$ under Assumption \ref{Ass:Exponential}. On the other hand, the estimate \eqref{est:hinfty} shows that 
	\begin{align*}
		|\nabla^j_k \left(\hat{\Gamma}(k,v_1,v_2) - k/|k|^2 (\nabla_{v_1}-\nabla_{v_2})(f f)(v_1,v_2)\right)|
		\leq \frac{C(K)}{1+|k|^{2+j}}, \quad \text{for $j=0,\ldots, 6$, $|k|\geq 1$}.
	\end{align*}
	Therefore, $\Gamma_2$ satisfies the estimate:
	\begin{align} \label{est:gamma2} 
		|\Gamma_2(x,v_1,v_2)|\leq \frac{Ce^{-(|v_1|+|v_2|)}}{|x| (1+|x|)^4}.
	\end{align} 
	Now inserting the estimates \eqref{est:gamma1} and \eqref{est:gamma2} into the representation
	\eqref{eq:gBconvrep} shows the estimates \eqref{est:gBMax} and \eqref{est:gBExp}.

	It remains to show that $g_B$ is in the space $W$ introduced in \eqref{def:W}). 
	We remark 
	that by construction $\hat{g}_B(k,v_1,v_2)=\hat{g}_B(-k,v_2,v_1)$, so $g_B$ satisfies the symmetry
	property \eqref{particlesymmetry}.  

	To show that $|h|[g_B] \in L^1_{loc}$ we use the decomposition \eqref{GammaDecomp}:
	\begin{align} \label{eq:gB1dec}
		g_B 			&= \frac{2\pi i}{|v_r|} (\Gamma_1+ \Gamma_2)(x,v_1,v_2) *_x \left( \cf_{(0,\infty)}(x \cdot \vr) \cdot \Haus^1 \llcorner \operatorname{span} \{\vr\}\right) =: g_{B,1} + g_{B,2}.
	\end{align}
	From the estimate \eqref{est:gamma2} we deduced that $g_{B,2}$ satisfies $|h|[g_{B,2}] \in L^1_{loc}$.
	
	We now estimate $|h|[g_{B,1}]$. To this end, we decompose the function further into:
	\begin{align} \label{eq:gBab}
		\hat{g}_{B,1}(k,v_1,v_2) =   \cf_{|\omega (v_1-v_2)|>1 }\hat{g}_{B,1}+ \cf_{|\omega (v_1-v_2)|\leq 1} \hat{g}_{B,1}  =: \hat{g}_{B,a} + \hat{g}_{B,b}.
	\end{align}
	Inserting the definition of $g_B$ \eqref{eq:Bogcorr}, and using $|v_1|\leq R$ we can estimate $g_{B,a}$  by:
	\begin{align*}
		\int_{\Reals^3} |g_{B,a}(x,v_1)| \ud{v_2} 
		\leq &C\left(1+ \int\int_{B_2}  |\nabla f(v_2)| |\hat{h}_B(k,v_1)|  \ud{k} \ud{v_2}  \right)+ \int_{B_2} \left|\int_{\Reals^3} \hat{h}_B(k,v) \ud{v} \right|\ud{k}
			  \\
		\leq &C(R) + C \int_{B_2}\left|\int_{\Reals^3}  e^{ikx}\hat{h}_B(k,v)\ud{v}\right| \ud{k} ,
	\end{align*}
	which is bounded by \eqref{eq:hbound}.
	Hence $|h|[g_{B,a}]\in L^1_{loc}$. 
	
	In order to estimate $g_{B,b}$ given by \eqref{eq:gBab}, we use the fact that $|\omega(v_1-v_2)|\leq 1$ and $|v_1|\leq R$ implies $|\omega v_2|\leq R+1$. 
	Hence $|\eps(k,-kv_2)|\geq c >0$ is bounded below uniformly on the support of $\hat{g}_{B,b}$, and
	$|h|[g_{B,b}] \in L^1_{loc}$ follows. Hence also $|h|[g_{B}] \in L^1_{loc}$ as claimed. 
	
	It then immediately follows that $h_B$ is indeed the marginal of $g_B$ (cf. \eqref{eq:Bogcorr}), since:
	\begin{align*}
		\int \hat{g}_B(k,v_1,v_2) \ud{v_2} = \int \frac{\hat{\phi}(k) \omega \left((\nabla_{v_1}-\nabla_{v_2} )(f  f) + \nabla f(v_1) \ol{\hat{h}}_B(k,v_2) - \nabla f(v_2)				\hat{h}_B(k,v_1)\right)}{\omega(v_1-v_2)-i0} \ud{v_2},
	\end{align*}
	and $\hat{h}_B$ satisfies the equation \eqref{eq:heq}. The estimates \eqref{est:hcoulM}-\eqref{est:hcoulE} imply  $\sup_{|v|\leq R}\|h[g_B](\cdot ,v)\|_{L^2} \leq C(R)$ as claimed.
\end{proofof}
\subsection{Soft potential interaction}

\begin{theorem}[Decay estimate for soft potentials] \label{thm:screensoft}
	We recall $g_B$ as introduced in Definition \ref{def:Bogcorr}, and assume $f$ satisfies
	the Assumptions \ref{ass:regdecay}-\ref{Ass:onepart1} and  $\phi=\phi_{S}$ is
	a soft potential (cf. Definition \ref{def:potentials}). Further we use the shorthand notation $v_r$, $\vr$ in \eqref{eq:vecshort}, 
	and $b,d,d_-$ introduced in \eqref{def:impactdist}.
	Write $v_r=v_1-v_2$, $\vr= v_r/|v_r|$ and let $\delta \in (0,1)$. 
	For almost every $(x,v_1) \in \Reals^3 \times \Reals^3$, there holds
	$g_B(z,v_1,\cdot) \in L^1(\Reals^3)$, and the marginal of $g$ coincides
	with $h_B$:
	\begin{align} \label{eq:hismarginalsoft}
	\int g_B(x,v_1,v_2) \ud{v_2} = h_B(x,v_1).
	\end{align} 
	Furthermore, for $n\in \Naturals$ the function $g_B$ satisfies the estimate: 
	\begin{align} 
	|g_B(x,v_1,v_2)| &\leq \frac{C(\delta)}{|v_r|}
	\frac{1}{(1+|b|+d_-)^{2-\delta}} e^{-(|v_1|+|v_2|)} \label{est:gBsoft}, \\
	|h_B(x,v_1)| 	&\leq \frac{C(\delta)} {1+|x|^{3-\delta}}e^{-(|v_1|+|v_2|)} . \label{est:hsoft}
	\end{align}
\end{theorem}
\begin{proof}
	The identity \eqref{eq:hismarginalsoft} follows analogously to the Coulomb case.
	For proving the estimates  \eqref{est:gBsoft}, \eqref{est:hsoft}, we recall the definition of $h$ in Fourier variables:
	\begin{align} 
	\hat{h}_B(k,v)		&:= f(v) \frac{(1-\eps(k,-kv))}{\eps(k,-kv)}- \hat{\phi}(k)	\frac{A^-(k,kv)}{\eps(k,-kv)} (\omega \nabla f(v)) .
	\end{align}
	Since $\eps$ is non-degenerate by Assumption, the functions $(1-\eps)/\eps$ and $A^-/\eps$ are bounded, as well as their first
	three derivatives in $k$. Using the exponential decay of $f(v)$ and $\nabla f(v)$, the decay estimate \eqref{est:hsoft} follows
	from Lemma \ref{cordecay}. A similar argument proves \eqref{est:gBsoft}.
\end{proof}
We observe that the result shows that the rate of decay is independent of the rate of the decay of the soft potential. Further, we do not
observe a singularity for small impact parameters $b$.

\section{Stability of the linearized evolution of the truncated two-particle correlation function} \label{sec:stability}

\subsection{The linearized evolution semigroup}

The goal of this subsection is to prove that the Bogolyubov propagator $\G$
introduced in Definition \ref{def:Bogolyubovprop} provides a strong
solution to the linear Bogolyubov evolution equation \eqref{BBGKYint3g}.
We start by proving the well-posedness of the propagator. Since the definition
involves the action of the Vlasov semigroup both on smooth initial data and on
Dirac masses, we first derive properties for both cases.
We recall that for translation invariant functions, we can reduce the number of variables
using \eqref{translationsystem}.

Since we prove the well-posedness of the linear evolution problem in the Schwartz space, we recall
the seminorms generating this space.
\begin{definition}\label{def:seminorm}
	For $k,l\in \Naturals_0$ and $n\in \Naturals$, let $\|\cdot\|_{C^{k,l}(\Reals^n)}$ be the seminorm defined by:
	\begin{align} \label{eq:seminormdef}
		\|f\|_{C^{k,l}(\Reals^n)}:= \sup_{x\in \Reals^n} (1+|x|)^l(|f(x)|+|\nabla^k f(x)|). 
	\end{align}
\end{definition} 
\begin{remark}
	The collection of norms $\|\cdot \|_{C^{k,l}(\Reals^n)}$ with $k,l\in \Naturals_0$ generates the Schwartz space, which
	can be equipped with the associated Frechèt-metric.
\end{remark}
\begin{lemma}[Solution of the Vlasov equation for Dirac masses] \label{lem:evdirac}
	Let $\phi=\phi_S$ be a soft potential, let $f\in \Sch(\Reals^3)$ satisfy Assumption \ref{Ass:onepart2} and let $x_0,v_0\in \Reals^3$. We set $h_0(x,v)=\delta(x-x_0) \delta(v-v_0)f(v)$. Consider the function $h(t)=\V(t)[h_0]$ defined by the Fourier-Laplace
	representation \eqref{hvlas}. Then there exists a function $Y\in C(\Reals^+,\Sch((\Reals^3)^3))$ such that
	$\partial_t Y(t,x) \in C(\Reals^+,\Sch((\Reals^3)^3))$ and: 
	\begin{align} \label{eq:evdir}
		h(t,x,v) = Y(t,x-x_0,v,v_0) + \delta(x-x_0-tv) \delta(v-v_0)f(v).
	\end{align} 
	Furthermore, $h$ is a weak solution to the Vlasov equation \eqref{linVlasov}, and $Y$ solves:
	\begin{align} \label{eq:Y}
	\partial_t Y + v\nabla_x Y - \nabla E_h \nabla f = 0, \quad Y(0,\cdot)=0.
	\end{align}
\end{lemma}
\begin{proof}
	We start by proving that $h$ can be decomposed as claimed in \eqref{eq:evdir}. W.l.og. let $x_0=0$. By the Fourier-Laplace representation of $h$ 
	in  \eqref{hvlas} we have:
	\begin{align} \label{eq:evdir1}
		\hat{h}(t,x,v) = \frac1{2\pi i} \int_{L_1} \tilde{h}(z,k,v) e^{zt} \ud{z} =  \frac1{2\pi i} \left(\int_{L_1}  \frac{\hat{h}_{0}(k,v)}{z+ikv} e^{zt}\ud{z}+ \int_{L_1} \frac{i Q(k,v) \tilde{\varrho}(z,k)}{z+ikv}e^{zt} \ud{z}\right) 
	\end{align} 	
	where $L_\gamma:= \{z \in \Complex: \Re(z)=\gamma \}$ is the line with real part $\gamma$, oriented upwards.
	The line integral is evaluated in the improper sense 
	\begin{align} \label{eq:invintegral}
		\int_{L_\gamma} f(z) \ud{z} = \lim_{T\rightarrow \infty}\int_{L_\gamma} f(z) \cf(|z|\leq T) \ud{z} . 
	\end{align}
	The first line integral in \eqref{eq:evdir1} is explicit and yields:
	\begin{align*}
		\frac1{2\pi i} \int_{L_1} \frac{\hat{h}_{0}(k,v)}{z+ikv}\ud{z} = e^{-ikv t} \hat{h}_0(k,v),
	\end{align*}
	so we obtain the second term in \eqref{eq:evdir}. It remains to show that the second line integral in \eqref{eq:evdir1} gives
	a function $Y$ with the desired properties. Using the formula \eqref{hvlas}, the term can be rewritten as:
	\begin{align} \label{eq:Yrepr}
		\hat{Y}(t,k,v,v_0) = \frac{f(v_0)}{(2\pi)^\frac32}\frac1{2\pi i}\int_{L_1}  \frac{iQ(k,v)  e^{zt}}{\eps(k,-iz)(z+ikv)(z+ikv_0)}\ud{z}	.
	\end{align}
	Now $\eps(k,-iz)$ is smooth and bounded below by Assumption \ref{Ass:onepart2}. The line integral is absolutely
	convergent and differentiating through it shows that for all $\ell_1,\ell_2, \ell_3 \in \Naturals_0$, $T>0$, there exists a $C>0$ such that:
	\begin{align}
		\|\nabla^{\ell_1}_{v_0} \nabla^{\ell_2}_{v}\nabla^{\ell_3}_{k} \frac1{2\pi i}\int_{L_1}  \frac{e^{zt}}{\eps(k,-iz)(z+ikv)(z+ikv_0)}\ud{z}\|_{C([0,T]\times \Reals^9)} \leq C.
	\end{align}
	Using that $Q$ and $f$ in \eqref{eq:Yrepr} are Schwartz functions, we obtain
	$Y \in C(\Reals^+, \Sch(\Reals^9))$. Next we observe that $\int h(t,x,v) \ud{v}= \varrho(t,x)$. To see this, we
	use  $\int \tilde{h}(z,k,v) \ud{v} = \tilde{\varrho}(z,k)$. The integration in $v$ commutes with the Laplace
	inversion \eqref{eq:evdir1}, so $\varrho$ is the spatial density of $h$. Hence the Fourier-Laplace definition \eqref{hvlas} of $h$ gives
	a weak solution of the Vlasov equation. Combining this with the decomposition \eqref{eq:evdir} we find that $Y$ is
	a weak solution to \eqref{eq:Y}. Using equation \eqref{eq:Y} we find $\partial_t Y \in C(\Reals^+, \Sch(\Reals^9))$ as claimed.
\end{proof}
\begin{lemma}[Vlasov equation with Schwartz initial data] \label{lem:evSchwartz}
	Let $\phi=\phi_S$ be a soft potential, let $f\in \Sch(\Reals^3)$ satisfy Assumption \ref{Ass:onepart2}. Further assume $h_0\in \Sch((\Reals^3)^2)$. Let  $h(t)=\V(t)[h_0]$ be defined by formula \eqref{hvlas}. There exists an $m \in \Naturals_0$ such that for any $k,l\in \Naturals_0$, there is a $C>0$ such that:
	\begin{align}\label{eq:vlasseminorm}
		\|h\|_{C^1([0,T];C^{k,l})} \leq C \|h_0\|_{C^{k+m,l+m}}.
	\end{align}
	 Further, the function is a strong solution to the Vlasov equation \eqref{linVlasov}.
\end{lemma} 
\begin{proof}
	For proving the estimate \eqref{eq:vlasseminorm}, we use the definition of $\V(t)[h_0]$ in Fourier-Laplace variables
	(cf. \eqref{hvlas}) to obtain the representation:
	\begin{align} \label{eq:evSchwartz1}
	\hat{h}(t,x,v) &=  \frac1{2\pi i} \left(\int_{L_1}  \frac{\hat{h}_{0}(k,v)}{z+ikv} e^{zt}\ud{z}+ \int_{L_1} \frac{i Q(k,v) \tilde{\varrho}(k,z)}{z+ikv}e^{zt} \ud{z}\right),  \\
	\tilde{\varrho}(k,z) & := \frac{\int\frac{\hat{h}_{0}(k,v^{\prime})}{%
			z+ikv^{\prime}}\ud{v^{\prime}}}{\eps(k,-iz)} .
	\end{align} 
	Since $\eps(k,-iz)$ is uniformly bounded below on the line $L_1$, the claim follows by differentiating through the integrals in	
	\eqref{eq:evSchwartz1}.
\end{proof}
We recall the Bogolyubov propagator $\G$ introduced in \eqref{eq:Bogolyubovprop}.
The previous two lemmas allow us to prove that the Bogolyubov propagator is well-defined. In order to
show that the function $g(t):= \G(t)[g_0]$ indeed solves the Bogolyubov equation, we show commutativity 
for Vlasov operators acting on different sets of variables. To this end we introduce the following shorthand notation.
\begin{notation}
	Let $S$ be the Schwartz distribution given by:
	\begin{align} \label{eq:Sdef}
	S(\xi_1,\xi_2) = \delta(\xi_1-\xi_2) f(v_1).
	\end{align}
\end{notation}
\begin{lemma} \label{lem:Vlasovcomm}
	Let $g_0(\xi_1,\xi_2) = \ol{g}_0(x_1-x_2,v_1,v_2)+ S(\xi_1,\xi_2)$, where $\ol{g}_0 \in \Sch$ and
	$S$ as introduced in \eqref{eq:Sdef}. Then the compositions of operators
	$\V_{\xi_1} \V_{\xi_2}[g_0]$, $\V_{\xi_2} \V_{\xi_1}[g_0]$ as introduced in Definition \ref{def:Bogolyubovprop} are well-defined and the following 
	commutation relation between $\V_{\xi_1}$ and $\V_{\xi_2}$ holds:
	\begin{align}\label{eq:Vlasovcomm}
		\V_{\xi_1}(t') \V_{\xi_2}(t)[g_0] &= \V_{\xi_2}(t) \V_{\xi_1}(t')[g_0].
	\end{align}
\end{lemma}
\begin{proof}
	By Lemma \ref{lem:evdirac}, $\V_{\xi_2}(t)[g_0]$ is the sum of a Schwartz function and a Dirac mass, so
	the composition with $\V_{\xi_1}(t')$ is well defined. The commutativity relation \eqref{eq:Vlasovcomm} follows from the explicit Fourier-Laplace
	representation \eqref{hvlas}.
\end{proof}
Now can now prove that $\G(t)$ gives the solution of the Bogolyubov equation \eqref{BBGKYint3g}. For convenience we introduce the following notation.
\begin{notation}
	We write $E_{j}[g]$, $j=1,2$ for the following expressions:
	\begin{align} \label{def:Ej}
		E_{2}[g](x,v_2) = \int \phi(x+y) g(y,v_1,v_2)\ud{v_1}, \quad E_{1}[g](x,v_1) = \int \phi(-x+y) g(y,v_1,v_2) \ud{v_2}.
	\end{align}
\end{notation}
\begin{theorem}[Solution of the linearized evolution equation] \label{thm:evolutionsol}
	Let $g_0$, $f$ be as in Theorem \ref{thm:stability}.
	The function $g$ given by  $g(t)=\G(t)[g_0]$ satisfies $g\in C(\Reals^+,\Sch((\Reals^3)^3))$, $\partial_t g\in C(\Reals^+,\Sch((\Reals^3)^3))$ and solves the Bogolyubov equation~\eqref{BBGKYint3g}. 
\end{theorem}
\begin{proof}
	First we observe that using the notation \eqref{def:Ej}, the Bogolyubov equation \eqref{BBGKYint3g} reads:
	\begin{equation}  \label{Bogeq}
		\begin{aligned} 
		\partial _\tau g+&(v_1-v_2) \nabla_x g- \nabla f(v_1) \nabla_x E_2[g](x,v_2)- \nabla f(v_2) \nabla_x E_1[g](x,v_1) 	 \\
		&=  (\nabla _{v_1} - \nabla_{v_2})\left(f(v_1)f(v_2)\right)\nabla \phi (x). 
		\end{aligned}
	\end{equation}
	We decompose $g(t)=\G(t)[g_0]$ into two parts:
	\begin{align}
		g(t)	&= \V_{\xi_1} \V_{\xi_2} [g_0]  + \left(\V_{\xi_1} \V_{\xi_2} [S] -T(t) S\right) = G_1 +G_2. 
	\end{align}
	We take the time derivative of both expressions. For the first term, the existence
	of the time derivative follows from Lemma \ref{lem:evSchwartz}, and using 
	Lemma \ref{lem:Vlasovcomm} we find:
	\begin{align}
		\partial_t G_1 = -\sum_{i \neq j} v_i \nabla_{x_i} G_1 + \nabla f(v_j) \nabla_{x_i} E_i[G_1].
	\end{align}
	To prove differentiability in time for $G_2$ we observe that
	\begin{align}
		G_2(t)= \V_{\xi_1}(t) [\V_{\xi_2}(t)[S]-T(t)S]+\left(\V_{\xi_1}(t) [S]-T(t)S\right)
	\end{align} 
	satisfies $G_2,\partial_t G_2\in C(\Reals^+,\Sch((\Reals)^9))$ by Lemma \ref{lem:evdirac}
	and Lemma \ref{lem:evSchwartz}. Differentiating $G_2$ yields:
	\begin{align}
		\partial_t G_2(t) = -\sum_{i\neq j} v_i \nabla_{x_i} G_2 + \nabla f(v_j) \nabla_{x_i} E_i[\V_{\xi_1} \V_{\xi_2}[S]]. 	
	\end{align}
	Now the claim follows from $\sum_{i\neq j=1}^2 \nabla f(v_j) E_i[T(t)[S]]=(\nabla_{v_1}-\nabla_{v_2})(f(v_1)f(v_2)) \nabla \phi(x)$.
\end{proof}

\subsection{Distributional stability of the Bogolyubov correlations} 
In Theorem \ref{thm:evolutionsol} we have proved that the Bogolyubov propagator	$\G(t)$ gives a solution to the Bogolyubov equation. 
In this subsection we prove the result \eqref{eq:distrstab} claimed in Theorem \ref{thm:stability}, that is the distributional stability of the Bogolyubov correlations. We split the problem into analyzing the solution $\Lambda$ of \eqref{Bogeq}  with non-zero initial datum $g_0$, but without the right-hand side in \eqref{Bogeq}, and the solution $\Psi$ of \eqref{Bogeq} with zero initial datum. 
The following lemma gives this decomposition in Fourier-Laplace variables.
\begin{lemma} \label{lem:gtdecomp}
	Let $g_0\in \Sch((\Reals^3)^3)$ be a function such that $g_0(x_1-x_2,v_1,v_2)$ is symmetric in exchanging 
	$\xi_1=(x_1,v_1)$, $\xi_2=(x_2,v_2)$. We make the decomposition
	\begin{align}
		g(t,\xi_1,\xi_2) &= \G(t)[g_0] = \Psi(t,t,\xi_1,\xi_2) + \Lambda(t,t,\xi_1,\xi_2),
	\end{align}
	where $\Psi(t,t',\xi_1,\xi_2) := \V_{\xi_1}(t) \V_{\xi_2}(t') [S]-T(t)[S]$, $\Lambda(t,t')=\V_{\xi_1}(t) \V_{\xi_2}(t') [g_0]$.
	Then the Fourier-Laplace representation of $\Psi$, written in the form
	\eqref{translationsystem}, satisfies:
	\begin{equation} \label{eq:psidecomp}
	\begin{aligned}
		\Psi(z,z',k,v_1,v_2) & := \Psi_1(z,z',k,v_1,v_2)
	+\Psi_2(z,z',k,v_1,v_2) +\Psi_2(z',z,-k,v_2,v_1) 
	\\
	\Psi_1(z,z',k,v_1,v_2) & := - \frac{Q(k,v_1) Q(-k,v_2)
		\int\frac{ \delta(v_1'-v_2') f(v_1')}{(z+ikv_1')(z'-ikv_2')} \ud{v_1'}\ud{v_2'}} {\eps(k,-iz) \eps(-k,-iz')
			(z+ikv_1) (z'-ikv_2)}   \\
		\Psi_2(z,z',k,v_1,v_2) & := \frac{ f(v_1)}{(z+ikv_1)} 	\frac{i Q(-k,v_2)}{\eps(-k,-iz')(z_2-ikv_1)(z'-ikv_2)}  
		\end{aligned}
	\end{equation}
	and the Fourier-Laplace representation of $\Lambda$ is given by:
	\begin{equation} \label{eq:lambdadecomp}
		\begin{aligned}
			\Lambda(z,z',k,v_1,v_2) & = \Lambda_1(z,z',k,v_1,v_2)	+\Lambda_2(z,z',k,v_1,v_2)	+\Lambda_2(z',z,-k,v_2,v_1)  \\
			\Lambda_1(z,z',k,v_1,v_2) & := \frac{g_0(k,v_1,-k,v_2)}{	(z+ikv_1)(z_2-ikv_2)} - \frac{Q(k,v_1) Q(-k,v_2) \int 		\frac{\hat{g}_0(k,v_1',-k,v_2')}{(z+ikv_1')(z'-ikv_2')} \ud{v_1'} \ud{v_2'}} {\eps(k,-iz_1) \eps(-k,-iz_2) (z_1+ikv_1)	(z_2-ikv_2)}   \\
			\Lambda_2(z,z',k,v_1,v_2) & := \frac{i Q(-k,v_2) \int \frac{	\hat{g}_0(k,v_1,-k,v')}{z'+ikv'}\ud{v'}}{\eps(-k,-iz_2)(z_1+ikv_1)(z_2-ikv_2)} . 
		\end{aligned}
	\end{equation}	
\end{lemma}
\begin{proof}
	Follows directly from the Fourier-Laplace representation of $\V$ in \eqref{hvlas} and
	the definition of the Bogolyubov propagator in Definition \ref{def:Bogolyubovprop}.
\end{proof}
We will start by proving two Lemmas that we will use throughout this whole section.
\begin{lemma}\label{lem:contint} 
	Let $H_\gamma =\{z \in \Complex: |\Re(z)|\leq \gamma\}$ and $f(k,z) \in L^1_{loc}(\Reals^3,\Complex)$,  such that there exist $R,c>0$ with $\|f(k,i \cdot)\|_{L^\infty(H_{c|k|})}\leq R$ for all $k\in \Reals^3$. Define the function 
	\begin{equation*}
		I(t,k,v,v'):= \int_{i \Reals-|c|k} \frac{e^{zt} f(k,iz)}{(z+ikv)(z+ikv')} \ud{z}.
	\end{equation*}
	Then for all $M,N\in \Naturals_0$, there exists $C>0$ such that 
	\begin{align}
		|\nabla^{M} _{v} \nabla^{N}_{v'} I(t,k,v,v')|\leq\frac{Ce^{-c|k|t}}{|k|} .  \label{contest}
	\end{align}
	Moreover, let $I$ be a function satisfying \eqref{contest} and $\kappa \in\Sch(\Reals^3)$ be a Schwartz function. Then for $p(k,v):= \PV \int
	\frac{\kappa(v')I(t,k,v,v')}{k(v-v')} \ud{v'}$ we have 
	\begin{align}  \label{contfinest}
		\|p(k,\cdot)\|_{C_{b}^1(\Reals^3)} \leq\frac{Ce^{-c|k|t}}{|k|}. 
	\end{align}
\end{lemma} 
\begin{proof}
	We start by proving \eqref{contest}. To this end, let $M,N\in \Naturals_0$ be arbitrary. Since $f$ is bounded on $H_{c|k|}$ , we can differentiate through the integral:
	\begin{align*}
		|\nabla^M _v \nabla^N_{v'} I(t,k,v,v')| \leq &e^{-c|k|t} \int_{i{\Reals}-c |k|} \frac{ |k|^{N+M} |f(k,iz)|}{|z+ikv|^{M+1}|z+ikv'|^{N+1}} \ud{z} \\
		\leq & C e^{-c|k|t} \int_{{\Reals}} \frac{ |k|^{N+M} }{(|k|+|r-kv|)^{M+1}(|k|+|r-kv'|)^{N+1}} \ud{r} \\
		\leq & C e^{-c|k|t}\int_{{\Reals}} \frac{ |k|^{N+M+1} }{%
			(|k|+|r|k|-kv|)^{M+1}(|k|+|r|k|-kv'|)^{N+1}} \\
		\leq &\frac{C e^{-c|k|t} } {|k|}\sup_{a,b \in{\Reals}} \int _{{\Reals}} \frac1{(1+|t-a|)^{M+1}(1+|t-b|)^{N+1}}\ud{t} 	\leq \frac{C e^{-c|k|t} } {|k|}.
		\end{align*}
	To prove \eqref{contfinest} we remark that $P(t,k,v,u):=\int I(t,k,v,v') \kappa(v')	\delta(kv'-u) \ud{v'}$ satisfies 
		\begin{equation*}
		|\nabla^{M} _{v} \nabla^{N}_{u} P(t,k,v,u)|\leq\frac{C e^{-c|k|t}}{|k|(1+|u|)^2}.
		\end{equation*}
		On the other hand $p(k,v)= \PV \int \frac{P(t,k,v,u')}{kv-u'} 
		\ud{u'}$ and the principal value integral can be bounded by 
		\begin{align*}
		|\PV \int \frac{P(u')}{u-u'} \ud{u'}| \leq C
		\left( \|P\|_{C^1}+ \|P\|_{L^1}\right) .
		\end{align*}
	\end{proof}
\begin{lemma}		\label{lem:dirac} Let $f \in \Sch( \Reals^3 \times\Reals^3)$ be a Schwartz function.
	\begin{enumerate}[(i)]
		\item For $t\rightarrow \infty$, the following convergence holds in the sense of Schwartz distributions:
			\begin{align} \label{eq:oscconv}
			\PV \frac{e^{-ik(v_1-v_2)t}}{k(v_1-v_2)} \longrightarrow-i \pi
			\delta(k(v_1-v_2))  \in \Sch'(\Reals^9).
			\end{align}
		\item For  $M\in{\mathbb{N}}_{0}$ arbitrary, the following convergence holds in $C^{M}_{b}(\Reals^3)$ as $t\rightarrow \infty$:
			\begin{align} \label{eq:osconvint}
			\PV \int f(k,v_2) \frac{e^{-ik(v_1-v_2)t}}{k(v_1-v_2)} \ud{k}
			\ud{v_2} \rightarrow-i \pi\int_{\Reals^3\times{\Reals}%
				^{3}} \delta(k(v_1-v_2)) f(k,v_2)\ud{v_2} \ud{k}.
			\end{align}
	\end{enumerate}
\end{lemma}
\begin{proof} We start by proving the convergence \eqref{eq:oscconv}. Let $w(k,v_1,v_2)$ be a Schwartz function and $W(k,u):= \int_{{\Reals}^{6}} \delta(k(v_1-v_2)-u) w(k,v_1,v_2) \ud{v_1} \ud{v_2}$. Let $\hat{W}$ be the Fourier transform in $u$, then: 
		\begin{align*}
		 & \PV\int \int_{\Reals^3 \times\Reals^3} \frac{e^{-ik(v_1-v_2)t}}{k(v_1-v_2)} w(k,v_1,v_2)\ud{v_1} \ud{v_2}\ud{k} 	=   \PV \int \int_{{\Reals}} \frac{e^{-iut}}{u} W(k,u) \ud{u} \ud{k}\\
		= & \int -i \sqrt{\frac{\pi}{2}} \sign(\xi+t)\hat{W}(k,\xi)\ud{\xi} \ud{k}	\rightarrow  -i \pi\int W(k,0) \ud{k}, \quad \text{as $t\rightarrow \infty$}.
		\end{align*}
		For proving \eqref{eq:osconvint}, we observe that $f\in \Sch$ implies that $F(k,u):= \int\delta(kv+u)f(k,v) \ud{v}$ is also Schwartz. Furthermore, we have 
		\begin{align*}
		& \PV \int f(k,v_2) \frac{e^{-ik(v_1-v_2)t}}{k(v_1-v_2)} \ud{%
			k} \ud{v_2} \\
		= & \int\PV \int_{{\Reals}} \frac{F(k,u)e^{-i(kv_1+u)t}}{kv_1+u} 
		\ud{u} \ud{k} 		=  \int\PV \int_{{\Reals}} \frac{F(k,u-kv_1)e^{-iut}}{u} \ud{u}
		\ud{k} \\
		\rightarrow  &\int F(k,k\cdot v_1) \ud{k},\quad \text{as $t\rightarrow \infty$}.
		\end{align*}
		Differentiating through the integral, we obtain the convergence for arbitrary
		derivatives in $v_1$.
\end{proof}
	\begin{lemma} \label{lem:stabzero}
		The solution $g(t)=\G(t)[N_0]$ to \eqref{BBGKYint3g} with zero initial datum $N_0:\equiv0$ converges to the
		Lenard solution in the sense of distributions, so 
		\begin{align*}
			\G(t) [N_0] \longrightarrow g_{B} \quad\text{ in $\Sch'(\Reals^9)$ as $t\rightarrow\infty$.}
		\end{align*}
	\end{lemma}
	\begin{proof}
		By Lemma \ref{lem:gtdecomp} we have $g(t,\cdot)=\G(t)[N_0](\cdot)=\Psi(t,t,\cdot)$. We use the Fourier-Laplace representation 
		$\Psi(z_1,z_2,k,v_1,v_2)  = \Psi_1(z_1,z_2,k,v_1,v_2)
		+\Psi_2(z_1,z_2,k,v_1,v_2) +\Psi_2(z_2,z_1,-k,v_2,v_1)$ in \eqref{eq:psidecomp}.
		We will show the distributional convergence term by term, starting with $\Psi_1$.		
\begin{lemma} \label{Psi1conv} 
	The following convergence holds in the sense of distributions: 
	\begin{align}
		\Psi_1(t,t,k,v_1,v_2)  		\longrightarrow & \frac{Q(k,v_1)Q(-k,v_2)}{%
		k(v_1-v_2)-i0}\int\frac{\frac{f(v')}{|\eps(k,-kv')|^2}}{k(v_1-v')-i0} \ud{v'}  \label{hv1} \\
			+ & \frac{Q(k,v_1)Q(-k,v_2)}{k(v_1-v_2)-i0}\int%
			\frac{\frac{f(v')}{|\eps(k,-k v')|^2}}{%
				k(v_2-v')-i0} \ud{v'}, \quad \text{as $t\rightarrow \infty$} .  \label{hv2}
	\end{align}
\end{lemma}
\begin{proof}
	First we perform the integration in $v_2'$ 
	\begin{align*}
		\Psi_1(z_1,z_2,k,v_1,v_2) & = - \frac{Q(k,v_1) Q(-k,v_2) \int
		\frac{ f(v')}{(z_1+ik v')(z_2-ik v')} \ud{v'} } {\eps	(k,-iz_1) \eps(k_2,-iz_2) (z_1+ik v_1) (z_2-ik v_2)}\\
			& = - \int\frac{Q(k,v_1) Q(-k,v_2) \frac{ f(v')}{(z_1+ik v')(z_2-ik v')} } {\eps(k,-iz_1) 
				\eps(k_2,-iz_2) (z_1+ik v_1) (z_2-ik v_2)} \ud{v'}.
	\end{align*}
	Now for $k$ fixed, we can perform the Laplace inversion
	integral both in $z_1$ and $z_2$. For $\Re(z_{i})>0$ the integrand has no singularities, so we can
	carry out the Laplace inversion on the contour with $\Re(z_i)=1$. By Assumption \eqref{ass:diel}, 
	$|\eps(k,-iz)|$ is bounded below for $\Re(z)=-ic|k|$ and some $c>0$. 
	The estimate \eqref{eq:Cauchy} allows to use Cauchy's residual theorem to move the contour to the left of the imaginary line:
	\begin{align*}
			& \frac1{2\pi i } \int_{i{\Reals}+c} \frac{Q(k,v)e^{zt}}{\eps%
				(k,-iz)(z+ikv)(z+ikv')} \ud{z} \\
			= & \frac1{2\pi i } \int_{i{\Reals}-c|k|} \frac{Q(k,v) e^{zt}}{\eps%
				(k,-iz)(z+ikv)(z+ikv')} \ud{z} + \PV \frac{Q(k,v)e^{-ikvt}}{\eps%
				(k,-kv)ik(v'-v)}+ \PV \frac{ Q(k,v)e^{-ikv't}}{\eps%
				(k,-kv')ik(v-v')} \\
			= & Q(k,v) \left( \frac1{2\pi i } \int_{i{\Reals}-c|k|} \frac{e^{zt}}{%
				\eps(k,-iz)(z+ikv)(z+ikv')} \ud{z} + \PV \frac{\frac{e^{-ikvt}}{%
					\eps(k,-kv)}-  \frac{e^{-ikv't}}{\eps(k,-kv')}}{ik(v'-v)}\right) \\
			=: & Q(k,v)(I(t,k,v,v') + R(t,k,v,v')).
	\end{align*}
	Writing $\Psi_1$ in terms of the functions $I$ and $R$ we obtain 
	\begin{align*}
			\Psi_1(t_1,t_2,k,v_1,v_2) & = -  \int
			f(v')Q(k,v_1)Q(-k,v_2)(I+R)(t,v_1,v')(I+R)(t,v_2,v') \ud{v'}.
	\end{align*}
	We expand the product $(I+R)(I+R)$ inside the integral. We claim all terms
	containing an integral term $I$ tend to zero in the limit $t\rightarrow\infty$ by
	Lemma \ref{lem:contint}. For the terms containing products of the form $IR$ this follows from 
	\eqref{contest}, for the products of the form $I I$ this can be inferred from 
	\eqref{contfinest} and the fact that the singularity in $k$ in estimate 
	\eqref{contfinest} is integrable. It remains to study the limiting
	behavior of the residual part: 
	\begin{align*}
		\Psi_1(t,k,v_1,v_2) + \int f(v')R(t,v_1,v')R(t,v_2,v') \ud{v'}
			\rightarrow0 \quad\text{ in $D'(\Reals^9)$}.
	\end{align*}
	In order to find the distributional limit of $\Psi_1$ we have to
	determine the limit of 
	\begin{align*}
		\Psi_{\infty}(t,k,v_1,v_2) := & - \int
			f(v')R(t,v_1,v')R(t,v_2,v') \ud{%
				v'} \\
			= & - Q(k,v_1) Q(-k,v_2) \PV \int f(v')\frac {%
				\frac{e^{-ik v_1t}}{\eps(k,-k v_1)}-\frac{e^{-ik v' t}}{\eps(k,-kv')}}{k(v'-v_1)} \frac{\frac {%
					e^{ik v_2t}}{\eps(-k,kv_2)}-\frac{e^{ik v't}}{\eps(-k,k v')}}{k(v'-v_2)} \ud{v'}.
	\end{align*}
	The denominator we split as 
	\begin{align}
			\frac1{k(v'-v_1)
				k(v'-v_2)}=\frac1{k(v_1-v_2)}\left( \frac1{k(v'-v_1)}-\frac1{k(v'-v_2)}\right) .
			\label{psiinfty}
	\end{align}
	Using this we can split $\Psi_{\infty}= \sum_{j=1}^2 \sum_{l=1}^{4} \Psi_{\infty}^{j,l}$, where $\Psi_{\infty}^{j,l}$ are given by (here $\zeta(1)=2$, $\zeta(2)=1$): 
	\begin{align*}
			\Psi_{\infty}^{j,1} (t,k,v_1,v_2) & := (-1)^{j}
			Q(k,v_1) Q(-k,v_2) \int f(v')\frac{\frac{e^{-ik (v_1-v_2)t}%
				}{\eps(k,-k v_1)\eps(-k,kv_2)}}{k(v'-v_j)k(v_1-v_2)} \ud{v'} \\
		\Psi_{\infty}^{j,2}(t,k,v_1,v_2) & :=(-1)^{j}
		Q(k,v_1) Q(-k,v_2) \int f(v')\frac{-\frac{e^{(-1)^{j}
					ik(v_{j}-v')t}}{\eps(k,-kv_1)\eps(-k,kv')}}{k(v'-v_{j})k(v_1-v_2)} \ud{v'} \\
		\Psi_{\infty}^{j,3}(t,k,v_1,v_2) & := (-1)^{j}
		Q(k,v_1) Q(-k,v_2) \int f(v')\frac{\frac1{\eps%
				(k_1,-k_1v')\eps(-k_1,k_1v')}}{k(v^{%
				\prime}-v_{j})k(v_1-v_2)} \ud{v'} \\
		\Psi_{\infty}^{j,4}(t,k,v_1,v_2) & := (-1)^{j}	Q(k,v_1) Q(-k,v_2) \int f(v')\frac{-\frac{e^{(-1)^{j}ik(v'-v_{\zeta(j))})t}}{\eps(k,-kv')\eps%
				(-k,kv_{\zeta(j))})}}{k(v'-v_{j})k(v_1-v_2)} 
		\ud{v'}.
	\end{align*}
	We compute the limits of these terms separately. Applying the Lemmas \ref{lem:contint} and \ref{lem:dirac} yields for $t\rightarrow \infty$:
	\begin{align*}
		\Psi_{\infty}^{j,1}(t,v_1,v_2) & \rightarrow(-1)^{j+1}\frac{i
		\pi\delta(k(v_1-v_2))}{\eps(k,-kv_1)\eps(-k,kv_2)%
		} Q(k,v_1) Q(-k,v_2) \PV \int \frac{f(v')}{k(v'-v_1)} \ud{v'} \\
		\Psi_{\infty}^{j,2}(t,v_1,v_2) & \rightarrow\frac{i \pi}{k(v_1-v_2)} Q(k,v_1) Q(-k,v_2) \int f(v')
		\frac{\delta(k(v'-v_j))}{|\eps(k,-kv')|^2} \ud{v'} \\
		\Psi_{\infty}^{j,3}(t,v_1,v_2) & \rightarrow(-1)^{j}\frac { Q(k,v_1) Q(-k,v_2)}{k(v_1-v_j)} \int \frac{%
		f(v')}{|\eps(k,-kv')|^2k(v'-v_1)} \ud{v'} \\
	\Psi_{\infty}^{j,4}(t,v_1,v_2) & \rightarrow0 \quad\text{for $%
		v_1\neq v_2$.}
	\end{align*}
	The terms $\Psi^{1,1}_{\infty}$ and $\Psi_{\infty}^{2,1}$ cancel. The
	remaining terms can be rearranged to:
	\begin{align*}
		\Psi_1(t,v_1,v_2) \rightarrow & \frac{Q(v
		_1)Q(-k,v_2)}{k(v_1-v_2)-i0} \int\frac{f(v')}{|\eps%
		(k,-k v')|^2k(v'-v_1)-i0} \ud{v'}\\
		+ &  \frac{Q(k,v_1)Q(-k,v_2)}{k(v_1-v_2)-i0}
		\int\frac{f(v')}{|\eps(k,-kv')|^2k(v'-v_2)-i0} \ud{v'}, \quad \text{as $t\rightarrow \infty$},
	\end{align*}
	 using Plemelj's formula.
\end{proof}

\begin{lemma}
	For $\Psi_2$ we have the following convergence in the sense of distributions:
	\begin{align*}
	\Psi_2(t,t,k,v_1,v_2) \rightarrow- & \frac{	f(v_1) Q(-k,v_2) }{\eps(-k,-kv_1)k(v_1-v_2)-i0}, \quad \text{as $t\rightarrow \infty$}.
	\end{align*}
\end{lemma}
\begin{proof}
	We argue similarly to the case of $\Psi_1$. We start from the
	definition of $\Psi_2$ 
	\begin{align*}
	\Psi_2(z_1,z_2,k,v_1,v_2) & = \frac{ f(v_1)}{%
		(z_1+ikv_1)} \frac{i Q(-k,v_2)}{\eps(-k,-iz_2)(z_2-ikv_1)(z_2-ikv_2)}
	\end{align*}
	and invert the Laplace transforms to obtain: 
	\begin{align*}
	\Psi_2(t_1,t_2,v_1,v_2) & = R(t_1,t_2,v_1,v_2)+ I(t_1,t_2,v_1,v_2) \\
	R(t_1,t_2,v_1,v_2) & :=	e^{-ikv_1t_1}f(v_1)Q(-k,v_2) \frac{\frac{%
			e^{ikv_2t_2}}{\eps(-k,-kv_2)}-\frac{e^{ikv_1t_2}}{\eps(-k,-kv_1)}}{-k(v_1-v_2)} \\
	I(t_1,t_2,v_1,v_2) & = f(v_1)e^{-ik v_1t_1}\frac1{2\pi i }	\int_{i{\Reals}-c|k|} \frac{i e^{z_2t_2} Q(-k,v_2)}{\eps(-k,-iz_2)(z_2-ikv_1)(z_2-ikv_2)} \ud{z_2}.
	\end{align*}
	We have $I(t,t,\cdot) \rightarrow0$ for $t\rightarrow \infty$, arguing as in the previous lemma.
	Hence we are left with the residual term $R$, which by
	Lemma \ref{lem:dirac} converges to 
	\begin{align*}
	R(t,t,v_1,v_2) = &
	e^{-ik v_1t}f(v_1)Q(-k,v_2) \frac{\frac{%
			e^{ikv_2t}}{\eps(-k,-kv_2)}-\frac{e^{ikv_1t}}{\eps%
			(-k,-kv_1)}}{-k(v_1-v_2)} \\
	\rightarrow & \delta(v_1-v_2)\frac{i \pi f(v_1)
		Q(-k,v_2) }{\eps(-k,-kv_2)} 	-  \frac{f(v_1) Q(-k,v_2) }{\eps(-k,-k v_1)k(v_1-v_2)},
	\end{align*}
	as $t\rightarrow \infty$. Using Plemelj's formula this proves the claim of the lemma.
\end{proof}
Combining the two previous lemmas, we obtain the following convergence in the sense of distributions:
\begin{align*}
g(t,v_1,v_2) \rightarrow & \frac{Q(k,v_1)Q(-k,v_2)}{k(v_1-v_2)-i0} \int\frac{f(v')}{|\eps(k,-kv')|^2k(v'-v_1)-i0} \ud{v'}\\
+ & \frac{ Q(k,v_1)Q(-k,v_2)}{k(v_1-v_2)-i0}
\int\frac{f(v')}{|\eps(k,-kv')|^2 k(v'-v_2)-i0} \ud{v'} \\
- &  \frac{f(v_1) Q(-k,v_2) }{\eps(-k,-kv_1)k(v_1-v_2)-i0} + \frac{f(v_2) Q(-k,v_2) }{\eps(-k,-kv_2)k(v_1-v_2)+i0},
\end{align*}
which by a rearrangement of terms coincides with $g_B$ (cf. \eqref{def:Bogcorr}). This finishes the proof
of Lemma \ref{lem:stabzero}.
\end{proof}
We now prove that the memory of the initial datum is erased by the evolution.
\begin{lemma}
	Let $g_0\in \Sch((\Reals^3)^3)$ be a function such that $g_0(x_1-x_2,v_1,v_2)$ is symmetric in exchanging $\xi_1$, $\xi_2$.
	Then the following holds:
	\begin{align*}
	\Lambda(t,t,x,v_1,v_2) =  \V_{\xi_1}(t)	\V_{\xi_2}(t)[g_{0}](x,v_1,v_2) \longrightarrow0 \quad\text{in $S'({\Reals}^{9})$ as $%
		t\rightarrow\infty$.}
	\end{align*}
\end{lemma}
\begin{proof}
	We start with the Fourier Laplace representation in \eqref{eq:lambdadecomp}: 
	\begin{align*}
	\Lambda(z_1,z_2,k,v_1,v_2) & =
	\Lambda_1(z_1,z_2,k,v_1,v_2)
	+\Lambda_2(z_1,z_2,k,v_1,v_2)
	+\Lambda_2(z_2,z_1,-k,v_2,v_1)
	\end{align*}
	The first term in $\Lambda_1$ is simply given by the action of the
	transport operator 
	\begin{align*}
	T(t)g_{0} (x,v_1,v_2) & = g_{0}(x-t(v_1-v_2),v_1,v_2).
	\end{align*}
	Since $g_{0} \in\s({\Reals}^{9})$, this term converges to zero in
	distribution. In the second term we perform the Laplace inversion, to
	split into a residual part and a contour integral left of the imaginary
	line: 
	\begin{align*}
	& \int_{\gamma_{c}} \int_{\gamma_{c}} \frac{e^{z_1t} e^{z_2t}Q(k,v_1)
		Q(-k,v_2) \int\int\frac{\frac12 g_{0}(k,v_1',-k,v_2')}{%
			(z_1+ikv_1')(z_2-ikv_2')} \ud{v_1'%
		} \ud{v_2'}} {\eps(k,-iz_1) \eps(-k,-iz_2)
		(z_1+ikv_1) (z_2-ikv_2)} \\
	= & Q(k,v_1)Q(-k,v_2) \int\int\frac12
	g_{0}(k,v_1',-k,v_2') (I+R)(t,k,v_1,v_1')
	(I+R)(t,-k,v_2,v_2') \ud{v_1'} \ud{	v_2'} \\
	& I(t,k,v,v') := \int_{\gamma_{-c|k|}} \frac{e^{zt}} {\eps%
		(k,-iz_1)(z_1+ikv_1) (z+ikv')}\ud{z} \\
	& R(t,k,v,v') := \frac{e^{-ikvt}}{\eps(k,-kv)}+\frac {%
		e^{-ikv't}}{\eps(k,-kv')}.
	\end{align*}
	Arguing as in the proof of Lemma \eqref{Psi1conv}, all terms containing an $I$
	converge to zero in distribution after expanding the product $(I+R)(I+R)$.
	The residual part $R$ converges to zero since $e^{i(v-w)t}\rightarrow0$ in $%
	\s'(\Reals^3 \times\Reals^3)$. The convergence $\Lambda_2 \rightarrow 0$ follows
	by an analogous computation.
\end{proof}

\subsection{Stability of the velocity fluxes}  

In this Subsection we prove the convergence result \eqref{eq:convfluxes} in Theorem \ref{thm:stability}. 
Consider the marginal $j(t,x,v_1):= \int g(t,x,v_1,v_2) \ud{v_2}$ of $g(t,\cdot)$. From \eqref{hvlas} 
we obtain the representation formula 
\begin{equation} \label{eq:jrepresentation}
	\begin{aligned}
		j(t,x,v_1) & = \psi(t,t,x_1-x_2,v_1) + \lambda(t,t,x_1-x_2,v_1)\\
		\psi(t,t,k,v_1) & = \psi_1(t,t,k,v_1) +\psi_2(t,t,k,v_1) -f(v_1)\\
		\psi_1(z_1,z_2,k,v_1) & := \frac{iQ(k,v_1) \int\int\frac {%
		\delta(v_1'-v_2')f(v_1')}{(z_1+ikv_1^{\prime})(z_2-ikv_2')}\ud{v_1'}\ud{%
		v_2'}}{(z_1+ikv_1)\eps(k,-iz_1)\eps(-k,-iz_2)} \\
		\psi_2(z_1,z_2,k,v_1) & := \frac{\int\frac{\delta(v_1-v_2^{%
			\prime})f(v_1)}{z_2-kv_2'}\ud{v_2'}}{%
		(z_1+ikv_1)\eps(-k,-iz_2)} \\
		\lambda(z_1,z_2,k,v_1) & = \lambda_1(z_1,z_2,k,v_1)
		+\lambda_2(z_1,z_2,k,v_1) \\
		\lambda_1(z_1,z_2,k,v_1) & := \int\frac{g_{0}(k,v_1,v_2)}{%
		(z_1+ikv_1)(z_2-ikv_2)} \ud{v_2}+ \frac{\frac12 \int \int%
		\frac{i Q(k,v_1) g_{0}(k,v_1',v_2')}{%
		(z_1+ikv_1')(z_2-ikv_2')} \ud{v_1'%
		} \ud{v_2'}} {\eps(k,-iz_1) \eps(-k,-iz_2)
		(z_1+ikv_1) } \\
		\lambda_2(z_1,z_2,k,v_1) & := \int\frac{ \frac12 \int\frac{ \hat {g}%
		_{0}(k,v_1,v')}{z_2+ikv'}\ud{v'}}{\eps%
		(-k,-iz_2)(z_1+ikv_1)} \ud{v_2}.
	\end{aligned}
\end{equation}
Further, we define the flux operator $J$ given by
\begin{align} \label{eq:fluxdef}
J[\psi](v_1) & := \nabla\cdot\left( \int- i k\hPhi(k) \psi(k,v_1) \ud{k}\right) .
\end{align}
\begin{lemma} \label{lem:fluxzero}
	The flux $J[\Psi]$ (cf. \eqref{eq:fluxdef}) converges to 
	\begin{align*}
	J[\psi](t,v_1) & \longrightarrow\nabla_{v_1} \left( \int\psi_{\infty
	}(k,v_1) \ud{k}\right) \quad\text{for all $v_1 \in \Reals^3$ as $t\rightarrow \infty$} \\
\psi_{\infty}(k,v_1) & := \int (\nabla_{v_1}-\nabla_{v'} f)(f  f)(v_1,v') \frac {%
	\delta(k(v_1-v'))(k \otimes k) |\hPhi(k)|^2}{|\eps(k,-kv_1)|^2} \ud{v'}.
\end{align*}
which is the velocity flux on the right-hand side of the Balescu-Lenard equation \eqref{eq:kinetic}.
\end{lemma}

\begin{proof}
	We show the convergence term by term, considering $J[\Psi_1]$, $J[\Psi_2]$ separately.
	Observe that $J[f(v_1)]=0$, since the function is independent of the space variable.
	 Let us first take a look at $\psi_2$. The
	integration in $v_2'$ can be carried out, and in the usual fashion
	we split the Laplace inversion in a contour integral left of the imaginary
	line and a residual: 
	\begin{align*}
	\psi_2(t,t,k,v_1) & = \frac{f(v_1)}{\eps(-k,kv_1)} +
	I(t,k,v_1), \quad 
	I(t,k,v_1)  := e^{-itkv_1} \int_{i{\Reals}-c|k|} \frac{e^{z_2t}}{%
		\eps(-k,-iz_2)(z_2-ikv_1)} \ud{z_2} .
	\end{align*}
	The contour integral vanishes in the limit $t\rightarrow \infty$, i.e. $J[I](t,v_1) \rightarrow0$. Therefore
	the contribution of $J[\psi_2]$ is 
	\begin{equation} \label{eq:last1}
		\begin{aligned}
			J[\psi_2] \rightarrow & - \nabla_{v_1} \left( \int ik \hPhi(k) \frac{%
			f(v_1)}{\eps(-k,kv_1)} \ud{k}\right) 	=  - \nabla_{v_1} \left( \int ik \hPhi(k) \frac{f(v_1)\eps(k,-kv_1)}{|%
			\eps(k,-kv_1)|^2} \ud{k}\right) \\
			= & - \nabla_{v_1} \left( \int(k \otimes k) |\hPhi(k)|^2 \frac {%
			\delta(k(v_1-v_1'))f(v_1) \nabla f(v_1')}{|\eps%
			(k,-kv_1)|^2} \ud{k}\right).
		\end{aligned}
	\end{equation}
	It remains to find the limit of $J[]\psi_1(t)]$. Again we can perform the
	integration in $v_2'$, obtaining 
	\begin{align*}
	\psi_1(z_1,z_2,k,v_1) & = \frac{iQ(k,v_1) \int\int\frac {%
			\delta(v_1'-v_2')f(v_1')}{(z_1+ikv_1^{%
				\prime})(z_2-ikv_2')}\ud{v_1'}\ud{%
			v_2'}}{(z_1+ikv_1)\eps(k,-iz_1)\eps(-k,-iz_2)} 	 = \frac{iQ(k,v_1) \int\frac{f(v_1')}{(z_1+ikv_1^{\prime
			})(z_2-ikv_1')}\ud{v_1'}}{(z_1+ikv_1)\eps%
		(k,-iz_1)\eps(-k,-iz_2)}.
	\end{align*}
	As in the previous lemmas, the Laplace inversion integral can be proved to be exponentially decaying in
	time up to a residual, which is given by
	\begin{align*}
	\lim_{t\rightarrow\infty} J[\psi_1] & = \lim_{t\rightarrow\infty}
	\nabla_{v_1} \cdot\left( \int k \hPhi(k) Q(k,v_1) \int f(v_1')
	R (t,k,v_1,v_1')\ud{v_1'} \ud{k}
	\right) \\
	R(t,k,v,v') & =\frac{e^{itkv'}}{\eps(-k,kv)} \left( \frac{%
		e^{-itkv}}{\eps(k,-kv)ik(v'-v)}-\frac{e^{-itkv'}}{\eps%
		(k,-kv')ik(v+v')}\right) .
	\end{align*}
	Applying Lemma \ref{lem:dirac}, we identify the limit as: 
	\begin{align} \label{eq:last2}
	\lim_{t\rightarrow\infty} J[\psi_1](t,v_1) & = \nabla_{v_1}
	\cdot\left( \int k \otimes k |\hPhi(k)|^2 \nabla f(v) \int\frac {%
		\delta(k(v_1-v_1'))f(v_1')}{|\eps(k,-kv_1')|^2}
	\ud{v_1'} \ud{k}\right) .
	\end{align}
	Summing \eqref{eq:last1} and \eqref{eq:last2}, we obtain as a limit of $J[\psi]$ 
	\begin{align*}
	 \lim_{t\rightarrow\infty} J[\psi] & = \nabla_{v_1} \cdot\left( \int (\nabla_{v_1}-\nabla_{v'} f)(f  f)(v_1,v')
	\delta(k(v_1-v')) \frac{k \otimes k |\hPhi(k)|^2)}{|\eps(k,-kv_1)|^2}
	\ud{k}\ud{v'} \right)
	\end{align*}
	as claimed.
\end{proof}
By a similar computation we obtain the following lemma.
\begin{lemma} \label{lem:fluxinitial}
	Let $J$ be the operator introduced in \eqref{eq:fluxdef}. For all $v_1\in \Reals^3$ there holds: 
	\begin{align*}
	J[\lambda](t,v_1) \longrightarrow0 \quad\text{as $t\rightarrow \infty$.}
	\end{align*}
\end{lemma} 
Combining Lemma \ref{lem:fluxinitial} with Lemma \ref{lem:fluxzero} shows the convergence of the velocity fluxes claimed in \eqref{eq:convfluxes}. This concludes the proof of Theorem \ref{thm:stability}.

\textbf{Acknowledgment.}
The authors acknowledge support through the CRC 1060
\textit{The mathematics of emergent effects}
at the University of Bonn that is funded through the German Science
Foundation (DFG).

\bibliographystyle{amsplain}
\bibliography{Bogolyubov}

\providecommand{\bysame}{\leavevmode\hbox to3em{\hrulefill}\thinspace}
\providecommand{\MR}{\relax\ifhmode\unskip\space\fi MR }
\providecommand{\MRhref}[2]{%
  \href{http://www.ams.org/mathscinet-getitem?mr=#1}{#2}
}
\providecommand{\href}[2]{#2}
\begin{thebibliography}{10}

\bibitem{balescu_statistical_1963}
R.~Balescu, \emph{Statistical mechanics of charged particles}, Monographs in
  {Statistical} {Physics} and {Thermodynamics}, {Vol}. 4, Interscience
  Publishers John Wiley \& Sons, Ltd., London-New York-Sydney, 1963.
  \MR{0160579}

\bibitem{balescu_equilibrium_1975}
\bysame, \emph{Equilibrium and nonequilibrium statistical mechanics},
  Interscience Publishers John Wiley \& Sons, Ltd., London-New York-Sydney,
  1975. \MR{0408641}

\bibitem{bobylev_particle_2013}
A.~Bobylev, M.~Pulvirenti, and C.~Saffirio, \emph{From particle systems to the
  {Landau} equation: a consistency result}, Comm. Math. Phys. \textbf{319}
  (2013), no.~3, 683--702. \MR{3040372}

\bibitem{bogoliubov_problems_1962}
N.~Bogoliubov, \emph{Problems of a dynamical theory in statistical physics},
  Studies in {Statistical} {Mechanics}, {Vol}. {I}, North-Holland, Amsterdam;
  Interscience, New York, 1962, pp.~1--118. \MR{0136381}

\bibitem{braun_vlasov_1977}
W.~Braun and K.~Hepp, \emph{The {Vlasov} dynamics and its fluctuations in the
  {$1/{N}$} limit of interacting classical particles}, Comm. Math. Phys.
  \textbf{56} (1977), no.~2, 101--113. \MR{0475547}

\bibitem{degond_spectral_1986}
P.~Degond, \emph{Spectral theory of the linearized {Vlasov}-{Poisson}
  equation}, Trans. Amer. Math. Soc. \textbf{294} (1986), no.~2, 435--453.
  \MR{825714}

\bibitem{desvillettes_polynomial_2015}
L.~Desvillettes, E.~Miot, and C.~Saffirio, \emph{Polynomial propagation of
  moments and global existence for a {Vlasov}-{Poisson} system with a point
  charge}, Ann. Inst. H. Poincaré Anal. Non Linéaire \textbf{32} (2015),
  no.~2, 373--400. \MR{3325242}

\bibitem{desvillettes_linear_1999}
L.~Desvillettes and M.~Pulvirenti, \emph{The linear boltzmann equation for
  long-range forces: a derivation from particle systems}, Math. Models Methods
  Appl. Sci. \textbf{09} (1999), no.~08, 1123--1145.

\bibitem{glassey_time_1994}
R.~Glassey and J.~Schaeffer, \emph{Time decay for solutions to the linearized
  {Vlasov} equation}, Transport Theory Statist. Phys. \textbf{23} (1994),
  no.~4, 411--453. \MR{1264846}

\bibitem{glassey_time_1995}
\bysame, \emph{On time decay rates in {Landau} damping}, Comm. Partial
  Differential Equations \textbf{20} (1995), no.~3-4, 647--676. \MR{1318084}

\bibitem{guernsey_kinetic_1962}
R.~Guernsey, \emph{Kinetic equation for a completely ionized gas}, Phys. Fluids
  \textbf{5} (1962), 322--328. \MR{0168315}

\bibitem{guo_landau_2002}
Y.~Guo, \emph{The {Landau} equation in a periodic box}, Comm. Math. Phys.
  \textbf{231} (2002), no.~3, 391--434. \MR{1946444}

\bibitem{krommes_two_1976}
J.~Krommes, \emph{Two new proofs of the test particle superposition principle
  of plasma kinetic theory}, Phys. Fluids \textbf{19} (1976), no.~5, 649--655.
  \MR{0416221}

\bibitem{lancellotti_fluctuations_2009}
C.~Lancellotti, \emph{On the fluctuations about the {Vlasov} limit for
  {${N}$}-particle systems with mean-field interactions}, J. Stat. Phys.
  \textbf{136} (2009), no.~4, 643--665. \MR{2540157}

\bibitem{lancellotti_glassey-schaeffer_2015}
\bysame, \emph{On the {Glassey}-{Schaeffer} estimates for linear {Landau}
  damping}, J. Comput. Theor. Transp. \textbf{44} (2015), no.~4-5, 198--214.
  \MR{3430537}

\bibitem{lancellotti_time-asymptotic_2016}
\bysame, \emph{Time-asymptotic evolution of spatially uniform {Gaussian}
  {Vlasov} fluctuation fields}, J. Stat. Phys. \textbf{163} (2016), no.~4,
  868--886. \MR{3488576}

\bibitem{lenard_bogoliubovs_1960}
A.~Lenard, \emph{On {Bogoliubov}'s kinetic equation for a spatially homogeneous
  plasma}, Ann. Physics \textbf{10} (1960), 390--400. \MR{0167274}

\bibitem{lifshitz_course_1981}
E.~Lifshitz and L.~Pitaevskii, \emph{Course of {Theoretical} {Physics}},
  Pergamon Press, Oxford, 1981.

\bibitem{marcozzi_derivation_2016}
M.~Marcozzi and A.~Nota, \emph{Derivation of the {Linear} {Landau} {Equation}
  and {Linear} {Boltzmann} {Equation} from the {Lorentz} {Model} with
  {Magnetic} {Field}}, J Stat Phys \textbf{162} (2016), no.~6, 1539--1565.

\bibitem{mouhot_landau_2011}
C.~Mouhot and C.~Villani, \emph{On {Landau} damping}, Acta Math. \textbf{207}
  (2011), no.~1, 29--201. \MR{2863910}

\bibitem{muskhelishvili_singular_1992}
N.~Muskhelishvili, \emph{Singular integral equations}, Dover Publications,
  Inc., New York, 1992. \MR{1215485}

\bibitem{nota_theory_2018}
A.~Nota, S.~Simonella, and J.~Velázquez, \emph{On the theory of {Lorentz}
  gases with long range interactions}, Rev. Math. Phys. \textbf{30} (2018),
  no.~03, 1850007.

\bibitem{oberman_theory_1983}
C.~Oberman and E.~Williams, \emph{Theory of fluctuations in plasma}, 1983,
  p.~111.

\bibitem{penrose_electrostatic_1960}
O.~Penrose, \emph{Electrostatic instabilities of a uniform non-{Maxwellian}
  plasma}, Phys. Fluids \textbf{3} (1960), no.~2, 258--265.

\bibitem{piasecki_stochastic_1987}
J.~Piasecki and G.~Szamel, \emph{Stochastic dynamics of a test particle in
  fluids with weak long-range forces}, Physica A Statistical Mechanics and its
  Applications \textbf{143} (1987), 114--122.

\bibitem{rostoker_superposition_1964}
N.~Rostoker, \emph{Superposition of {Dressed} {Test} {Particles}}, The Physics
  of Fluids \textbf{7} (1964), no.~4, 479--490.

\bibitem{spohn_kinetic_1980}
H.~Spohn, \emph{Kinetic equations from {Hamiltonian} dynamics: {Markovian}
  limits}, Rev. Modern Phys. \textbf{52} (1980), no.~3, 569--615. \MR{578142}

\bibitem{spohn_large_2012}
\bysame, \emph{Large {Scale} {Dynamics} of {Interacting} {Particles}}, Springer
  Science \& Business Media, 2012 (en).

\bibitem{strain_linearized_2007}
R.~Strain, \emph{On the linearized {Balescu}-{Lenard} equation}, Comm. Partial
  Differential Equations \textbf{32} (2007), no.~10-12, 1551--1586.
  \MR{2372479}

\bibitem{velazquez_non-markovian_2018}
J.~Velázquez and R.~Winter, \emph{From a non-{Markovian} system to the
  {Landau} equation}, Comm. Math. Phys., to appear (2018).

\bibitem{villani_review_2002}
C.~Villani, \emph{A review of mathematical topics in collisional kinetic
  theory}, Handbook of mathematical fluid dynamics, vol.~1, North-Holland,
  Amsterdam, 2002.

\end{thebibliography}

\end{document}